\documentclass[a4paper,11pt]{article}

\usepackage[T1]{fontenc}
\usepackage[utf8]{inputenc}
\usepackage{fullpage}
\usepackage{times}
\usepackage{microtype}
\usepackage{amsmath,amssymb,amsthm,mathtools}
\usepackage{booktabs}
\usepackage{tabularx}
\usepackage{float}
\usepackage{enumitem}
\usepackage[small]{caption}
\usepackage{authblk}
\usepackage{natbib}
\usepackage{xcolor}
\usepackage{algorithm}
\usepackage[noend]{algpseudocode}
\usepackage[hidelinks]{hyperref}
\usepackage[capitalize,nameinlink]{cleveref}

\setlist[enumerate,1]{label=(\roman*),leftmargin=2.2em,itemsep=0.25em,topsep=0.4em}
\usepackage{thmtools}
\makeatletter
\IfFormatAtLeastTF{2026-06-01}{%
  \renewcommand\thmt@autorefsetup{%
    \@xa\def\csname\thmt@envname autorefname\@xa\endcsname\@xa{\thmt@thmname}%
  }%
}{}
\makeatother

\newtheorem{theorem}{Theorem}[section]
\newtheorem{proposition}[theorem]{Proposition}
\newtheorem{lemma}[theorem]{Lemma}
\newtheorem{corollary}[theorem]{Corollary}
\theoremstyle{definition}

\theoremstyle{remark}

\newcommand{\eps}{\varepsilon}

\newcommand{\E}{\mathbb{E}}
\newcommand{\Prb}{\mathbb{P}}
\newcommand{\one}{\mathbf{1}}
\newcommand{\calA}{\mathcal{A}}

\newcommand{\calE}{\mathcal{E}}

\newcommand{\DA}{\mathsf{DA}}
\newcommand{\Like}{\mathsf{Like}}
\newcommand{\Rand}{\mathsf{Rand}}

\title{\bf Online Fair Division under an Oblivious Adversary}
\author[]{Saar Cohen}
\author[]{Nicholas Teh}
\author[]{Michael Wooldridge}
\affil[]{University of Oxford, UK}
\date{\vspace{-10mm}}

\begin{document}
\maketitle

\begin{abstract}
We study online fair division, where indivisible goods arrive one at a time and must be allocated immediately and irrevocably among $n\ge2$ agents. Against an oblivious adversary that fixes the input in advance, we ask what fairness guarantees randomized online algorithms can achieve. We present the \emph{Deficit Allocation} (DA) algorithm, which improves the previously known high-probability approximation to proportionality up to one good (PROP1) from inverse-logarithmic to inverse-doubly-logarithmic in $n/\delta$, where $\delta$ is the failure probability. DA does not require $\delta$ as input. With probability at least $1-\delta$, a refined variant achieves an arbitrarily close approximation to proportionality when no single item is too valuable relative to an agent's proportional share; without any restriction on item values, it obtains an arbitrarily close approximation to proportionality up to $k$ goods (PROP$k$) for $k$ doubly logarithmic in $n/\delta$. For $n2^{-2n}\le\delta\le1/2$, our bounds on the approximation factor and $k$ are asymptotically tight for the stronger fractional benchmark that divides each item equally among the agents who value it positively. We also prove that algorithms that use only which agents value each item positively, but not by how much, cannot improve upon the logarithmic dependence of uniformly random allocation. However, other fairness notions admit weaker guarantees. For envy-freeness up to one good (EF1), every randomized online algorithm has an input on which any positive approximation succeeds with probability at most $e^{-\Omega(n)}$. For envy-freeness up to any good (EFX), this probability is at most $1/n!$ with only $n+1$ items, tight within a factor of $(n+1)/2$; for the maximin share (MMS), it can be at most $5/6$. Finally, allocating each item uniformly among the agents who value it positively gives, with probability at least $1-\delta$, an arbitrarily close approximation to envy-freeness up to $k$ goods (EF$k$) for $k$ logarithmic in $n/\delta$, and this dependence is necessary for this rule.
\end{abstract}

\section{Introduction}
\label{sec:introduction}

Consider a compute cluster that shares its GPUs among several machine learning teams
\citep{mahajanEtAl2020themis}. Blocks of spare capacity become available unpredictably, and each
block must be assigned to a team as soon as it appears: an idle accelerator is wasted capacity,
and preempting a training run once it has started is costly. The teams do not value a block
equally, and the scheduler cannot know which blocks will become available next. Similar
decisions arise when a recommender system chooses which content producer receives the exposure
generated by each incoming user session \citep{patroEtAl2020fairrec}, or when a ride-hailing
platform assigns each ride request to a driver before the next request is known
\citep{suhrEtAl2019two}. In each case, the platform would like to treat its agents (e.g., teams,
producers, or drivers) fairly, yet it must commit to every decision without seeing the future.

The \emph{online fair division} model captures these settings
\citep{aleksandrovEtAl2015online,aleksandrovWalsh2020survey}. Indivisible goods arrive one at a
time; when a good arrives, agents' values for it are revealed, and an online algorithm must
allocate it immediately and irrevocably. The model was originally motivated by a food bank
that distributes donations as they arrive \citep{aleksandrovEtAl2015online}. To evaluate
the resulting allocation, we consider
standard fairness notions. A classical criterion is \emph{proportionality}
\citep{aleksandrovWalsh2020survey}, which requires each of the $n$ agents to receive at least a
$1/n$ fraction of her value for all goods. With indivisible goods, exact proportionality may be
impossible: if there is a single good, for example, only one agent can receive it. A standard
relaxation is \emph{proportionality up to one good} (PROP1), which requires every agent to reach
her proportional share after adding one good from outside her bundle. When all goods are known
in advance, PROP1 allocations always exist \citep{conitzerEtAl2017public}. Online, however, even
PROP1 is demanding: when allocating a good, the algorithm knows neither how many goods remain nor
how valuable they will be.

How difficult this is depends on the adversary choosing the input. If the adversary is
\emph{adaptive}, choosing each good after observing how the earlier goods were allocated, no
online algorithm can guarantee any positive approximation to PROP1 that is independent of the
number of goods, and this remains true for \emph{proportionality up to $k$ goods} (PROP$k$) for
every fixed $k$ \citep{neohTeh2026closing}. A factor that depends on the number $m$ of goods is
attainable \citep{duSun2026online}, but any such factor necessarily vanishes as $m$ grows. One
way around this impossibility is to give the algorithm information about the future, such as
every agent's total value \citep{neohEtAl2026additional} or her largest value for a single good
\citep{duSun2026online,neohTeh2026closing} (\cref{subsec:related-work}). We instead require no information at all: we
weaken the adversary and allow the algorithm to randomize. An \emph{oblivious} adversary may
choose any input, even a worst-case one, but must fix the entire input in advance. Randomization
is essential here: the choices of a deterministic algorithm can be computed ahead of time, so an
oblivious adversary can fix in advance exactly the input that an adaptive adversary would
generate, and the impossibility above therefore applies to every deterministic algorithm.

The oblivious adversary model fits the applications above well: spare capacity appears when
production workloads finish, users open an app, and riders request trips, all because of demand
outside the platform, which may follow any pattern but cannot react to the platform's random
choices. By contrast, arrivals that respond to realized allocations, for example through feedback
loops in recommendation, are better captured by an adaptive adversary. Unlike stochastic
models of arrivals \citep{benadeEtAl2024fair,gaoEtAl2021market}, the oblivious model makes no
distributional assumptions, which is valuable when the arrival process is nonstationary or
poorly understood.

With randomization, fairness can be evaluated either ex ante or in the realized allocation. The weaker notion is \emph{ex-ante} proportionality, which requires each agent's \emph{expected} utility to be at least her proportional share. This is easy to achieve (e.g., by giving every good to a uniformly random agent), but provides little protection in the realized allocation, where some agent may receive far less than her share. For example, a team may receive no computing resources even under a policy that is fair in expectation. We therefore focus on \emph{realized} guarantees. For $\alpha\in[0,1]$, an allocation is \emph{$\alpha$-PROP1} if every agent obtains at least an $\alpha$ fraction of her proportional share after adding one good from outside her bundle. We seek a large approximation factor $\alpha$, either with probability at least $1-\delta$, where $\delta\in(0,1)$ is the failure probability, or in expectation. In the latter case, we take the minimum over agents \emph{before} the expectation, so the guarantee measures the fairness of the worst-off agent in the same realized allocation.

Against an oblivious adversary, \citet{chooEtAl2026approximate} showed that allocating each good
to a uniformly random agent, which we call \emph{uniformly random allocation}, returns a
$\Theta(1/\log(n/\delta))$-PROP1 allocation with probability at least $1-\delta$, and that this
dependence on $\delta$ is tight for that rule. Since the rule makes every decision
independently of the past, it cannot compensate an agent who has been unlucky so far, which
suggests that an algorithm with memory might do better. Beyond PROP1, one would also like to
know whether almost the entire proportional share is attainable, and whether any of this extends
to envy-based fairness or to the maximin share. This leads to the question at the core of this
paper:
\begin{quote}
\emph{Against an oblivious adversary, to what extent can randomized online algorithms guarantee fairness?}
\end{quote}

\subsection{Our Results}
\label{subsec:contributions}

Our results reveal a strong separation between fairness notions. For proportionality,
an algorithm using allocation history improves the dependence on $n/\delta$ from inverse-logarithmic
to inverse-doubly-logarithmic compared with uniformly random allocation. For envy-freeness and the maximin share, by contrast, no algorithm achieves even an
arbitrarily small positive approximation with probability close to one on every input.
Approximate envy-freeness becomes attainable when logarithmically many goods may be removed.
All our algorithms are fully online: they know $n$, but not the number or values of future
goods or any agent's total value. \Cref{tab:results} summarizes the main guarantees.

\paragraph{Approximate PROP1 (\cref{sec:main-result}).}
We present \emph{Deficit Allocation} ($\DA$), which uses allocation history to favor agents who
have received little relative to their shares. For every $\delta\in(0,1)$, $\DA$ returns an
$\Omega(1/\log\log(n/\delta))$-PROP1 allocation with probability at least $1-\delta$
(\cref{thm:main}), improving the logarithmic dependence of uniformly random allocation to doubly
logarithmic; it also improves on the concurrent guarantees of \citet{chenLong2026randomized}
(\cref{subsec:related-work}). The same algorithm works for every $\delta$ without knowing it and achieves
$\Omega(1/\log\log n)$ in expectation. If each
good is valued positively by at most $\kappa$ agents, the
ratio improves to $\Omega(\min\{1,n/(\kappa\log\log(n/\delta))\})$, without knowing $\kappa$.
Using allocation history alone is insufficient: algorithms that know which agents value each
good positively but ignore value magnitudes cannot improve the logarithmic dependence over a
range of parameters (\cref{thm:value-magnitudes}). Unlike uniformly random allocation and the algorithm of \citet{chenLong2026randomized}, $\DA$
is not ex-ante proportional (\cref{prop:not-ex-ante}).

\paragraph{Near-proportionality and PROP$k$ (\cref{sec:extensions}).}
We use the stronger fractional benchmark $\mu_i$ that divides each good equally among its
positive valuers, and let $\nu_i$ denote agent $i$'s maximum single-good value. For every fixed
$\eps\in(0,1)$, a variant of $\DA$ guarantees
$(1-\eps)\mu_i-O_\eps(\nu_i\log\log(n/\delta))$ for every agent with probability at least
$1-\delta$ (\cref{thm:near-proportional}). The subtracted term does not depend on the
number of goods. Thus, if each proportional share is at least
$K_\eps\nu_i\log\log(n/\delta)$, where $K_\eps$ depends only on $\eps$, then every agent receives a $(1-\eps)$ fraction of her proportional share with the same probability
(\cref{cor:small-goods}); uniformly random allocation requires a factor of order $\log(n/\delta)$ \citep{chooEtAl2026approximate}. Without restrictions on values, the
variant achieves $(1-\eps)$-PROP$k$ for $k=O_\eps(\log\log(n/\delta))$, while $\DA$ achieves
$\Omega(\min\{1,k/\log\log(n/\delta)\})$-PROP$k$ for every $k$, in both cases with
probability at least $1-\delta$ and without being given $k$ (\cref{thm:propk}). For $\mu_i$, these bounds are optimal when
$n2^{-2n}\le\delta\le1/2$, but this does not settle optimality for PROP1
(\cref{thm:fractional-lower}). The guarantees also extend to unequal entitlements
(\cref{cor:weighted}).

\paragraph{Lower bounds for EF1, EFX, and MMS (\cref{sec:lower-bounds}).}
An allocation is \emph{envy-free up to one good} (EF1) if no agent envies another after some
good is removed from the latter's bundle \citep{liptonEtAl2004approximately,budish2011competitive}. Although EF1 allocations always exist offline
\citep{liptonEtAl2004approximately}, for every $\alpha\in(0,1]$ and randomized online algorithm,
some input limits the probability of $\alpha$-EF1 to $e^{-\Omega(n)}$, even with only $O(n)$
goods and their number known in advance (\cref{thm:ef1-confidence}). Thus, the worst-case expected EF1 approximation ratio of every algorithm is at most $e^{-\Omega(n)}$ (\cref{cor:ef1-expected}), compared with
$\Omega(1/\log\log n)$ for PROP1 under $\DA$. For \emph{envy-freeness up to any good} (EFX),
which requires envy-freeness after any single good is removed \citep{caragiannisEtAl2019unreasonable}, only $n+1$ goods suffice to
reduce the success probability to $1/n!$, tight within a factor of $(n+1)/2$
(\cref{thm:efx-confidence}); with this many goods, PROP1 is achievable deterministically. The
EFX probability remains exponentially small even when every agent's maximum value for a single good
is known (\cref{thm:efx-known-max}). Finally, the \emph{maximin share} (MMS) is the value an
agent can guarantee by partitioning the goods into $n$ bundles and receiving her least valuable
one \citep{budish2011competitive}. Extending a two-agent construction of \citet{zhouEtAl2023multiagent}, we show that some
input limits the probability of any positive MMS approximation to
$1-\lfloor n/2\rfloor/(2n)\le5/6$ (\cref{thm:mms-confidence}). Conversely, the small-good
condition of \cref{cor:small-goods} guarantees $(1-\eps)$-MMS.

\paragraph{Approximate EF$k$ (\cref{sec:efk}).}
An allocation is \emph{envy-free up to $k$ goods} (EF$k$) if no agent envies another after at
most $k$ goods are removed from the latter's bundle. Allowing these additional removals changes
the picture: for every $\eps\in(0,1)$, the $\Like$ rule
\citep{aleksandrovEtAl2015online} returns a $(1-\eps)$-EF$k$ allocation with probability at
least $1-\delta$ for $k=O(\eps^{-1}\log(n/\delta))$, independently of the number and values of
the goods (\cref{thm:efk}). $\Like$ is also exactly ex-ante envy-free. Its logarithmic dependence
on $n/\delta$ is necessary over a range of parameters, even for binary valuations
(\cref{prop:like-efk-lower}), and its worst-case expected EF$k$ ratio is
$1-\Theta(\log n/k)$ for $k\ge\log n$ (\cref{cor:efk-expected}).

\subsection{Related Work}
\label{subsec:related-work}

Online fair division was introduced by \citet{aleksandrovEtAl2015online}; see
\citet{aleksandrovWalsh2020survey} for a survey of online fair division and
\citet{amanatidisEtAl2023fair} for a survey of the fair division of indivisible goods more
broadly. We focus here on the work most closely related to ours. In
\cref{app:further-related-work}, we discuss further related models, including offline
allocation with both ex-ante and realized guarantees, unequal entitlements, restricted
valuations and stochastic arrivals, online models with reassignments or subsidies, online welfare
maximization and maximin shares, temporal voting, and repeated fair division. There, we also compare our deficit bound with online
correlated selection, whose induction argument we extend, and with balanced allocation
processes.

\paragraph{Randomized online fair division.}
Our work builds most directly on that of \citet{chooEtAl2026approximate}, who establish
high-probability PROP1 guarantees for uniformly random allocation against an oblivious
adversary and show that their dependence on $\delta$ is tight for that rule, though not for all
randomized online algorithms. \Cref{thm:value-magnitudes} shows that, for
$2ne^{-n/8}\le\delta\le1/4$, no algorithm that ignores value magnitudes can improve this
logarithmic dependence. Their small-good analysis also gives near-proportional utilities when
each agent's proportional share is logarithmically larger than her largest value for a single
good. We improve both bounds by correlating allocations across arrivals and using value
magnitudes. \citet{neohTeh2026closing} prove the impossibility against adaptive adversaries
that motivates our model, and analyze the $\Like$ rule, which allocates uniformly among agents
with positive values, against an oblivious adversary. Since $\DA$ uses the same fractional
allocation as its benchmark, it retains the advantage of $\Like$ when few agents value each
good positively.

\paragraph{Concurrent work.}
In concurrent and independent work, \citet{chenLong2026randomized} also study realized fairness
against an oblivious adversary. Their algorithm, Layered Scale Quota, returns an
$\Omega(\sqrt{\log n}/\log(n/\delta))$-PROP1 allocation with probability at least $1-\delta$,
and its expected realized PROP1 guarantee is $\Omega(1/\sqrt{\log n})$. \Cref{thm:main}
improves these bounds to $\Omega(1/\log\log(n/\delta))$ and $\Omega(1/\log\log n)$,
respectively. On the other hand, their algorithm is exactly ex-ante proportional, whereas $\DA$
is not (\cref{prop:not-ex-ante}). For envy-based fairness, they show that, for every
$\alpha\in(0,1]$, some input limits the probability of returning an $\alpha$-EF1 allocation to
$(n+1)/(2n)$, and that the worst-case expected realized EFX guarantee is zero, even with
identical valuations. \Cref{thm:ef1-confidence} strengthens their EF1 bound from a constant to
$e^{-\Omega(n)}$. Their EFX construction allows an unrestricted number of goods, whereas
\cref{thm:efx-confidence} bounds the success probability with exactly $n+1$ goods, and our
three-value construction (\cref{thm:efx-known-max}) also allows the number of goods and all
maximum values to be known.

\paragraph{Information about the future.}
Advance information permits stronger deterministic guarantees. \citet{neohEtAl2026additional}
obtain exact PROP1 when every agent's total value is known, and show that offline guarantees for
share-based notions carry over when every agent's multiset of values is known. For known maximum
values, \citet{duSun2026online} and \citet{neohTeh2026closing} independently obtain
$1/2$-PROP1, and \citet{neohTeh2026closing} improve this to $19/30$-PROP1. Without any
information, \citet{duSun2026online} give a deterministic $\Omega(1/\log(nm))$-PROP1 guarantee
against an adaptive adversary, without knowing $m$; by \citet{neohTeh2026closing}, any such
factor must vanish as $m$ grows. Our algorithms require none of this information, and our
guarantees are independent of $m$. For EFX, \citet{melissourgosProtopapas2026predictions} study
predictions of future valuations when total values are known. Their impossibility without
predictions already uses $n+1$ goods for $n\ge3$, but chooses later goods according to earlier
allocations; our bound of $1/n!$ quantifies the success probability of arbitrary randomized
algorithms on a fixed input, and \cref{thm:efx-known-max} gives an exponential bound even when
maximum values are known. When future values and their order are known,
\citet{elkindEtAl2025temporal} study EF1 after every arrival, and \citet{choi2026tfdmm} extend
the analysis to the case of \emph{mixed manna} (i.e., both goods and chores). A different
relaxation allows earlier allocations to be changed: \citet{heEtAl2019future} study how many
such changes are needed to maintain EF1, and \citet{huangTeh2026upgrade} obtain MMS guarantees
when goods may be reassigned in ways that leave every affected agent no worse off (see
\cref{app:further-related-work} for further models with reassignments). Our allocations, by
contrast, are irrevocable.

\paragraph{Additive envy and deterministic guarantees.}
A complementary line of work measures envy additively rather than multiplicatively.
\citet{benadeEtAl2024fair} study online envy minimization and its compatibility with efficiency
under different assumptions on the arriving goods, and \citet{halpernEtAl2025envy} relate online
additive envy minimization to multicolor discrepancy. Their lower bound concerns additive envy
with inverse-polynomial probability; it does not show that an $\alpha$-EF1 allocation can be
returned only with exponentially small probability, and hence does not imply
\cref{thm:ef1-confidence}. Our positive EF$k$ result also concerns a different quantity: it
compares utilities multiplicatively after removing actual goods from the envied bundle, and an
additive envy bound in units of the largest single-good value does not by itself bound the
value that can be removed using $k$ goods from that bundle. \citet{kahanaEtAl2026perpetual}
also use deficits to obtain deterministic fairness guarantees at every arrival for both item
allocation and public decisions. For proportionality, their additive loss after $m$ arrivals is
$O(\log n+\sqrt{m\log n/n})$ times the largest value for a good outside the agent's bundle.
Their EF$k$ guarantee assumes a known finite set of possible values, and the number of removals
grows with $m$. By allowing a fixed multiplicative loss and a prescribed error probability, we
obtain bounds that are independent of $m$ and require no advance information about the possible
values. When all goods and their values are known in advance,
\citet{goldbergEtAl2026minimizing} study the computational complexity of minimizing the sum, over all arrivals, of the maximum pairwise additive envy.

\section{Preliminaries}
\label{sec:preliminaries}

For a positive integer $k$, let $[k]=\{1,\dots,k\}$. We write $O_\eps(\cdot)$ for bounds
whose hidden constants depend only on $\eps$, and use $\log\log x$ to denote
$\max\{1,\ln\ln x\}$. We consider $n\ge2$ agents with nonnegative additive valuations $v_i$
over indivisible goods. An input is a finite sequence $I=(g_1,\dots,g_m)$ of goods arriving
online. Upon the arrival of $g_t$, the value vector $(v_i(g_t))_{i\in[n]}$ is revealed, and the
algorithm must assign $g_t$ immediately and irrevocably; goods cannot be discarded. A
\emph{fully online} algorithm knows $n$ in advance, but not $m$, future goods, or their values.
We consider an oblivious adversary, which fixes the entire input before the algorithm makes any
random choices. Moreover, an algorithm is \emph{non-wasteful} if, whenever some agent values an
arriving good positively, it assigns the good to an agent who values it positively.

For each arriving good $g_t$, let $N_t=\{i\in[n]:v_i(g_t)>0\}$ denote the set of agents who
value it positively. Let $\kappa:=\max_t|N_t|\le n$ be the maximum number of agents who
value any single good positively; algorithms need not know $\kappa$. We assume
$\kappa\ge1$; otherwise, all agents value every good at zero and every allocation satisfies
the fairness notions below. Let $M=\{g_1,\dots,g_m\}$ be the set of all goods. For each agent
$i$, let $V_i=v_i(M)$ denote her total value for all goods and
$\nu_i=\max(\{v_i(g):g\in M\}\cup\{0\})$ her maximum value for a single good. An allocation
$A=(A_1,\dots,A_n)$ assigns bundle $A_i$ to each agent $i$, giving her utility $v_i(A_i)$ and
proportional share $V_i/n$.

\paragraph{Proportionality and the fractional benchmark.}
For an allocation $A=(A_1,\dots,A_n)$, let
$b_i^{(k)}=\max\{v_i(S):S\subseteq M\setminus A_i,\ |S|\le k\}$
denote the largest value agent $i$ can obtain from at most $k$ goods outside her bundle, and
write $b_i=b_i^{(1)}$. An allocation is \emph{$\alpha$-PROP$k$} if
$v_i(A_i)+b_i^{(k)}\ge\alpha V_i/n$ for every agent $i$. For $k=1$, this is
$\alpha$-PROP1. Here and below, $\alpha\in[0,1]$, and we omit $\alpha$ when it equals one.
The outside goods are added only for the fairness comparison, not by the allocation algorithm.

We also consider a fractional benchmark that divides each good equally among the agents who
value it positively. Specifically, let $x_{it}=1/|N_t|$ if $i\in N_t$ and $x_{it}=0$ otherwise.
Agent $i$'s utility under this fractional allocation is
$\mu_i=\sum_t x_{it}v_i(g_t)$. Since $|N_t|\le\kappa$ whenever $i\in N_t$, we have
$\mu_i\ge V_i/\kappa\ge V_i/n$. Thus, any approximation to $\mu_i$ also approximates agent
$i$'s proportional share, while providing a stronger benchmark when fewer agents value each
good positively.

\paragraph{Envy-freeness and MMS.}
For distinct agents $i$ and $j$, let $b_{ij}^{(k)}=\max\{v_i(S):S\subseteq A_j,\ |S|\le k\}$ denote the largest value, according to agent $i$, of at most $k$ goods in agent $j$'s bundle. An allocation is \emph{$\alpha$-EF$k$} if $v_i(A_i)\ge\alpha(v_i(A_j)-b_{ij}^{(k)})$ for every pair $i\ne j$; the removed goods may depend on the pair. For $k=1$, this is $\alpha$-EF1 \citep{liptonEtAl2004approximately,budish2011competitive}. The stronger notion of \emph{$\alpha$-EFX} requires $v_i(A_i)\ge\alpha v_i(A_j\setminus\{g\})$ for every pair $i\ne j$ and every good $g\in A_j$ \citep{caragiannisEtAl2019unreasonable}. Our EFX lower bounds use strictly positive values and therefore also apply when only positively valued goods may be removed.

Agent $i$'s \emph{maximin share}, denoted $\mathsf{MMS}_i$, is the value she can guarantee herself by partitioning all goods into $n$ bundles and receiving her least valuable bundle. Formally, it is given by $\mathsf{MMS}_i=\max_{(P_1,\dots,P_n)}\min_{j\in[n]}v_i(P_j)$, where the maximum is over all partitions $(P_1,\dots,P_n)$ of $M$ \citep{budish2011competitive}. An allocation is \emph{$\alpha$-MMS} if $v_i(A_i)\ge\alpha\mathsf{MMS}_i$ for every agent $i$. In particular, $\mathsf{MMS}_i\le V_i/n$.

\paragraph{Realized and ex-ante guarantees.}
For any fairness notion $F$ defined above, let $\rho_F(A)$ denote the \emph{realized approximation ratio} of allocation $A$, namely, the largest $\alpha\in[0,1]$ for which $A$ is $\alpha$-$F$. We treat a zero fairness benchmark as satisfied. For example, the realized PROP$k$ ratio is $\rho_{\mathrm{PROP}k}(A)=\min_{i:V_i>0}\min\{1,n(v_i(A_i)+b_i^{(k)})/V_i\}$. For a randomized algorithm $\calA$ and input $I$, let $\rho_F(\calA,I)$ denote the resulting random realized ratio, and let $R_F^{\calA}(n)=\inf_I\E[\rho_F(\calA,I)]$ denote its worst-case expected realized guarantee over inputs with $n$ agents. Importantly, the minimum over agents is taken before the expectation, so the ratio evaluates all agents in the same realized allocation. All of our probabilistic guarantees likewise concern simultaneous fairness of the final allocation: the \emph{success probability} for $\alpha$-$F$ on input $I$ is $\Prb[\rho_F(\calA,I)\ge\alpha]$.

By contrast, ex-ante guarantees compare expected utilities separately. An algorithm is \emph{ex-ante proportional} if $\E[v_i(A_i)]\ge V_i/n$ for every agent $i$ and every input, and \emph{ex-ante envy-free} if $\E[v_i(A_i)]\ge\E[v_i(A_j)]$ for every pair $i,j$ and every input, where the expectations are over the algorithm's randomness. Thus, ex-ante fairness need not ensure that all agents are approximately fair in the same realized allocation.

\begin{table}[t]
\caption{Main probabilistic guarantees, with $\eps,\delta\in(0,1)$ and positive integer $k$.
Upper bounds hold on every fixed input; lower bounds state that, for every randomized online
algorithm and $\alpha\in(0,1]$, some fixed input satisfies the stated probability bound, even
when the number of goods is known. All guarantees concern the final allocation.}
\label{tab:results}
\small
\renewcommand{\arraystretch}{1.15}
\setlength{\tabcolsep}{4pt}
\begin{tabularx}{\textwidth}{@{}lll>{\raggedright\arraybackslash}X@{}}
\toprule
Fairness notion & Approximation factor & Success probability & Algorithm or input conditions \\
\midrule
\multicolumn{4}{@{}l}{\emph{Guarantees on every fixed input}} \\
PROP1 & $\Omega\left(\min\left\{1,\frac{n}{\kappa\log\log(n/\delta)}\right\}\right)$
& $\ge1-\delta$ & $\DA$; no additional conditions (\cref{thm:main}). \\
Proportionality & $1-\eps$ & $\ge1-\delta$
& A variant of $\DA$; $V_i/n\ge K_\eps\nu_i\log\log(n/\delta)$ for all $i$;
$K_\eps$ depends only on $\eps$ (\cref{cor:small-goods}). \\
PROP$k$ & $\Omega\left(\min\left\{1,\frac{nk}{\kappa\log\log(n/\delta)}\right\}\right)$
& $\ge1-\delta$ & $\DA$; any $k$ (\cref{thm:propk}). \\
PROP$k$ & $1-\eps$ & $\ge1-\delta$
& A variant of $\DA$; $k\ge K_\eps\log\log(n/\delta)$ (\cref{thm:propk}). \\
EFX & $1$ & $\ge\frac{2}{(n+1)n!}$
& $m=n+1$; uniform over allocations with one two-good bundle and all other bundles
singletons (\cref{thm:efx-confidence}). \\
EF$k$ & $1-\eps$ & $\ge1-\delta$
& $\Like$; $k\ge\lceil(2/\eps)\ln(n(n-1)/\delta)\rceil$ (\cref{thm:efk}). \\
\midrule
\multicolumn{4}{@{}l}{\emph{Probability upper bounds for every randomized online algorithm}} \\
EF1 & Any $\alpha\in(0,1]$ & $\le e^{-\Omega(n)}$
& $m=O(n)$; four positive values; sufficiently large $n$ (\cref{thm:ef1-confidence}). \\
MMS & Any $\alpha\in(0,1]$ & $\le1-\frac{\lfloor n/2\rfloor}{2n}$
& $m=2n-\lfloor n/2\rfloor$; values in $[0,1]$ (\cref{thm:mms-confidence}). \\
EFX & Any $\alpha\in(0,1]$ & $\le1/n!$
& $m=n+1$; strictly positive values (\cref{thm:efx-confidence}). \\
EFX & Any $\alpha\in(0,1]$ & $\le O(n^{5/2}2^{-n})$
& $m=n+1$; three positive values; maximum values $\nu_i=1$ may be known
(\cref{thm:efx-known-max}). \\
\bottomrule
\end{tabularx}
\end{table}

\section{Approximate PROP1 via Deficit Allocation}
\label{sec:main-result}

As discussed in \Cref{sec:introduction}, uniformly random allocation achieves
$\Theta(1/\log(n/\delta))$-PROP1 with probability at least $1-\delta$, but allocates each good
independently of earlier outcomes. In particular, it cannot compensate agents who have received less than their fractional amounts. Our \emph{Deficit Allocation} ($\DA$) algorithm addresses this by using the allocation history. For each agent and range of values,
it tracks how much of her fractional amount has not yet been matched by received goods and
favors agents with larger deficits.

Goods of very different values should not be treated alike: receiving several low-valued goods
should not compensate an agent for missing a much more valuable one. We therefore partition each
agent's positive values into intervals $[2^d,2^{d+1})$, called \emph{scales}. Within a scale,
values differ by less than a factor of two, so the number of goods captures their total value up
to this constant factor.

For each agent--scale pair $(i,d)$ encountered during the run, $\DA$ creates a nonnegative
\emph{deficit} $D_{i,d}$, starting at zero. Intuitively, this deficit records the accumulated
fractional amount at that scale that has not yet been compensated by goods received by agent
$i$. Larger deficits should make an agent more likely to receive future goods at the same scale.
To implement this preference, $\DA$ assigns a deficit $y$ the weight
$w(y)=\exp(H(y)/3)$, where $H(y)=(4/3)^y-1$, and biases its allocation probabilities toward
larger weights. The rapidly increasing weight function is chosen so that large deficits become
doubly exponentially unlikely, which is the key probabilistic property used in our analysis
below. The particular constants are technical choices for establishing this tail bound.

\begin{algorithm}[htb]
\caption{Deficit Allocation ($\DA$)}
\label{alg:da}
\begin{algorithmic}[1]
\Require Number of agents $n\ge2$.
\State When an agent--scale pair $(i,d)$ first occurs, set $D_{i,d}=0$.
\For{each arriving good $g_t$}
    \If{$N_t=\varnothing$}
         allocate $g_t$ arbitrarily and continue.
    \EndIf
    \For{each $i\in N_t$}
        \State Set $d_i=\lfloor\log_2 v_i(g_t)\rfloor$ and $x_{it}=1/|N_t|$.
    \EndFor
    \State Choose $j\in N_t$ with probability
    $\displaystyle\frac{x_{jt}w(D_{j,d_j})}{\sum_{i\in N_t}x_{it}w(D_{i,d_i})}$.
    \State Allocate $g_t$ to $j$.
    \For{each $i\in N_t$}
        \State $D_{i,d_i}\gets\max\{0,D_{i,d_i}+x_{it}-6\one\{i=j\}\}$.
    \EndFor
\EndFor
\end{algorithmic}
\end{algorithm}

Each good increases the relevant deficit of every positive valuer by its fractional share $x_{it}$. If agent $j$ receives the good, her deficit is reduced by six, but not below zero, allowing one received good to compensate for a constant amount of accumulated fractional mass. Any remaining deficit persists and influences future allocations. Since the update occurs after allocation, the probabilities in \cref{alg:da} depend only on deficits accumulated before the current good arrives. The algorithm is fully online and non-wasteful, and stores one deficit per encountered agent--scale pair.
For example, suppose two agents value successive goods at the same scale. Their deficits are initially zero, so the first good is allocated uniformly. The recipient's deficit remains zero, while the other's becomes $1/2$, making her more likely to receive the next good. Over time, deficits track accumulated unmatched fractional mass rather than only the previous allocation. Separate scales prevent many low-valued goods from compensating for much more valuable ones.
\begin{restatable}{theorem}{thmmain}
\label{thm:main}
For every $\delta\in(0,1)$ and every input, with probability at least $1-\delta$, $\DA$
returns an $\alpha$-PROP1 allocation, where
\begin{equation*}
    B_{n,\delta}=\frac{\ln(1+2\ln(2n/\delta))}{\ln(4/3)} \, \text{ and } \, \alpha=\min\left\{1,\frac{n}{\kappa(12+4B_{n,\delta})}\right\}.
\end{equation*}
In particular, $\alpha=\Omega(\min\{1,n/(\kappa\log\log(n/\delta))\})$, so the allocation is
$\Omega(1/\log\log(n/\delta))$-PROP1.
\end{restatable}

Since $\DA$ uses neither $\delta$ nor $\kappa$, $\DA$  satisfies the theorem
for every $\delta$ and benefits from a small $\kappa$ without knowing it. Taking
$\delta=1/2$ also gives $R^{\DA}_{\text{PROP1}}(n)=\Omega(1/\log\log n)$.
If PROP1 is required for all but a fixed positive fraction of the agents, the approximation
factor becomes constant with probability $1-e^{-\Omega(n)}$ (\cref{prop:all-but}).

The improvement from a logarithmic to a doubly logarithmic dependence on $n/\delta$ stems from $\DA$ making large deficits doubly exponentially unlikely. This phenomenon is reminiscent of the power of two choices in balls-into-bins processes, where favoring less loaded bins reduces the maximum load from nearly logarithmic to doubly logarithmic \citep{azarEtAl1999balanced}; see \cref{app:further-related-work} for a comparison. We next show how this tail bound translates into the PROP1 guarantee; the joint probability bound and its parameter checks appear in \cref{app:deficits,app:simple-parameters}, while \cref{app:main-proof} gives the full derivation. Let $D_{i,d}$ denote the final deficit. The bound applies to any set $S$ of encountered agent--scale pairs:
\[
 \Prb[D_{i,d}\ge y_{i,d}\text{ for all }(i,d)\in S]
 \le e^{-\sum_{(i,d)\in S}H(y_{i,d})}
\]
for nonnegative thresholds $y_{i,d}$. This joint bound assumes no independence and is needed because agents compete for the same goods: a large deficit increases an agent's weight, while the competing weights are also random. It also covers an agent who receives a good but retains a large deficit; her previous deficit must then have been larger by at least five, making the event sufficiently unlikely.

\begin{proof}[Proof Sketch of \cref{thm:main}]
Let $A$ be the returned allocation. For each agent--scale pair $(i,d)$, let $Y_{i,d}$ be the total fractional amount accumulated at that pair and $N_{i,d}$ the number of goods at scale $d$ received by agent $i$. All sums below range over the scales encountered by agent $i$. Summing the deficit updates yields $Y_{i,d}\le6N_{i,d}+D_{i,d}$, since truncation at zero can only increase the deficit relative to the untruncated update. Since goods at scale $d$ have value in $[2^d,2^{d+1})$ to agent $i$, it follows that
\begin{equation}
 \mu_i\le12v_i(A_i)+2\sum_d 2^dD_{i,d}.
 \label{eq:main-accounting}
\end{equation}

\cref{lem:joint-deficit} implies the marginal bound $\Prb[D_{i,d}\ge y]\le e^{-H(y)}$. For the convex increasing function $\phi(y)=e^{H(y)/2}$, integration gives $\E[\phi(D_{i,d})]\le2$. Further, every encountered scale satisfies $2^d\le\nu_i$, and thus $\sum_d2^d\le2\nu_i$. Applying convexity with weights $2^d/(2\nu_i)$, and assigning any remaining weight to zero, gives
\[
 \E\left[\phi\left(\frac{\sum_d2^dD_{i,d}}{2\nu_i}\right)\right]\le2
 \quad \text{for } \nu_i>0.
\]
Since $H(B_{n,\delta})=2\ln(2n/\delta)$, Markov's inequality gives $\Prb[\sum_d2^dD_{i,d}>2B_{n,\delta}\nu_i]\le\delta/n$. Agents with $\nu_i=0$ also have $\mu_i=0$. A union bound over the agents, together with \eqref{eq:main-accounting}, thus gives, with probability at least $1-\delta$,
\begin{equation}
 \mu_i\le12v_i(A_i)+4B_{n,\delta}\nu_i\quad\text{for every }i.
 \label{eq:main-utility}
\end{equation}
For each agent, a most valuable good either belongs to her bundle or lies outside it, so $\nu_i\le v_i(A_i)+b_i$. Combining this with $\mu_i\ge V_i/\kappa$ shows that the allocation is $\min\{1,n/(\kappa(12+4B_{n,\delta}))\}$-PROP1. Finally, $B_{n,\delta}=O(\log\log(n/\delta))$ gives the claimed asymptotic bound.
\end{proof}

The proof uses no union bound over scales or arrivals. Their number may be arbitrarily large and is not known in advance, yet the loss depends only on $n$, $\delta$, and the agent's maximum single-good value. Further, \eqref{eq:main-utility} is stronger than PROP1: when an agent's fractional utility $\mu_i$ is sufficiently large relative to $B_{n,\delta}\nu_i$, she receives a constant fraction of $\mu_i$ without adding an outside good.

\subsection{Why value magnitudes matter}
\label{subsec:value-magnitudes}

Taking earlier allocations into account is not sufficient by itself. The next result shows
that an algorithm must also use value magnitudes to improve on uniformly random allocation.

\begin{restatable}{theorem}{thmvaluemagnitudes}
\label{thm:value-magnitudes}
Let $n\ge48$ and $2ne^{-n/8}\le\delta\le1/4$. Set
$L=\lfloor\log_4(n/(4\delta))\rfloor$. For every randomized online algorithm $\calA$ whose
choice at time $t$ depends only on $N_1,\dots,N_t$, its earlier allocations, and its internal
randomness, there is an input $I$ with $Ln$ goods and values in $\{1,1/(Ln)\}$ such that
$\Prb[\rho_{\text{PROP1}}(\calA,I)\le8/L]>\delta$.
\end{restatable}

The proof in \cref{app:value-magnitudes} chooses each value independently from two positive
numbers. The algorithm cannot distinguish them, and some agent receives none of her
valuable goods with probability greater than $\delta$. We next improve the fraction of
the benchmark attained by $\DA$ to any fixed factor below one.

\section{Near-Proportionality and \texorpdfstring{PROP$k$}{PROPk}}
\label{sec:extensions}

The analysis of $\DA$ loses constant factors in two places: values within the same scale may
differ by a factor of almost two, and receiving a good reduces the corresponding deficit by
six. Neither loss is intrinsic to the approach. We next modify $\DA$ to use narrower scales
and smaller deficit reductions, allowing it to preserve an arbitrarily large fraction of the
fractional benchmark.

\subsection{Preserving almost the entire fractional share}
\label{subsec:near-proportional}

For $\eps\in(0,1)$, let $\calA_\eps$ denote this refined variant of $\DA$. It replaces the
factor-two scales by narrower scales in which values differ by less than $1+\eps/2$. To permit
a smaller deficit reduction while retaining the doubly exponential tail bound, each
agent--scale pair is represented by $\lceil8/\eps\rceil$ parallel deficits, among which each
fractional amount is divided equally. When a good is allocated, only the selected copy of the
recipient's deficit is reduced, by $1+\eps/2$, while the recipient receives the entire good.
The algorithm and its parameters appear in \cref{app:near-proportional}. It depends on
$\eps$, but not on $\delta$ or future values.

\begin{restatable}{theorem}{thmnearproportional}
\label{thm:near-proportional}
For every fixed $\eps\in(0,1)$, every $\delta\in(0,1)$, and every input, with probability
at least $1-\delta$, $\calA_\eps$ returns an allocation $A$ in which every agent $i$ has
utility $v_i(A_i)\ge(1-\eps)\mu_i-O_\eps(\nu_i\log\log(n/\delta))$.
\end{restatable}

The narrower scales and smaller deficit reductions replace the coefficient $12$ in
\eqref{eq:main-accounting} by $(1+\eps/2)^2$. Applying the same convexity argument across
copies and scales then yields \cref{thm:near-proportional}; the full analysis appears in
\cref{app:near-proportional}. Importantly, the additive loss is independent of the number of
goods, and the guarantee applies to agents' actual utilities, without adding outside goods.
This yields near-proportionality whenever no single good is too large relative to an agent's
proportional share.

\begin{restatable}{corollary}{corsmallgoods}
\label[corollary]{cor:small-goods}
For every $\eps\in(0,1)$, there is a constant $K_\eps>0$ such that, for every
$\delta\in(0,1)$ and every input satisfying
$V_i/n\ge K_\eps\nu_i\log\log(n/\delta)$ for every agent $i$, with probability at least
$1-\delta$, $\calA_{\eps/2}$ gives every agent $i$ utility
at least $(1-\eps)V_i/n$.
\end{restatable}

To see this, apply \cref{thm:near-proportional} with accuracy $\eps/2$. Under the condition
above, its additive loss is at most $\eps V_i/(2n)$ for every agent
(\cref{app:small-goods}). For fixed $\eps$, this improves the small-good condition for uniformly
random allocation \citep{chooEtAl2026approximate} from logarithmic to doubly logarithmic in
$n/\delta$. The same result also guarantees $(1-\eps)$-MMS, since
$\mathsf{MMS}_i\le V_i/n$. We do not optimize the dependence on $\eps$.

\subsection{Proportionality up to \texorpdfstring{$k$}{k} goods}
\label{subsec:propk}

The small-good condition in \cref{cor:small-goods} can be avoided by allowing more outside
goods. As $k$ grows, $\DA$ achieves a stronger PROP$k$ approximation, while its refined
variant approaches the full fractional benchmark once $k$ is doubly logarithmic in $n/\delta$.

\begin{restatable}{theorem}{thmpropk}
\label{thm:propk}
For every positive integer $k$, every $\delta\in(0,1)$, and every input, the following hold:
\begin{enumerate}
\item With probability at least $1-\delta$, $\DA$ returns an $\alpha$-PROP$k$ allocation with
$\alpha=\Omega(\min\{1,\allowbreak nk/(\kappa\log\log(n/\delta))\})$.
\item For every $\eps\in(0,1)$, if $k\ge K_\eps\log\log(n/\delta)$, then with probability
at least $1-\delta$, $\calA_{\eps/2}$ returns an allocation satisfying
$v_i(A_i)+b_i^{(k)}\ge(1-\eps)\mu_i$ for every agent $i$.
\end{enumerate}
Here, $K_\eps$ is the constant in \cref{cor:small-goods}. Neither algorithm is given $k$.
\end{restatable}

For $k=1$, the first statement recovers \cref{thm:main}. Since $\mu_i\ge V_i/n$, the second
statement gives $(1-\eps)$-PROP$k$ on unrestricted inputs for
$k=O_\eps(\log\log(n/\delta))$. Moreover, if $\kappa\le(1-\eps)n$, then
$(1-\eps)\mu_i\ge V_i/n$, so the same guarantee implies exact PROP$k$.

The dependence on $k$ requires more than simply dividing the additive error in
\eqref{eq:main-utility} by $k$, since one outside good may be much more valuable than the
others. Instead, for each agent $i$, we consider her $k$th largest value $a_i$ and treat
separately the fewer than $k$ goods lying above its scale. The remaining scales incur an error
proportional to $a_i$, while the $k$ largest goods ensure
$ka_i\le v_i(A_i)+b_i^{(k)}$. Combining these observations yields the improvement with $k$;
the full arguments appear in \cref{app:propk-proof,app:propk-refined}.

\subsection{Optimality for the fractional benchmark}
\label{subsec:fractional-lower}

The guarantees above are stated against the fractional benchmark $\mu_i$, which can be strictly
larger than the proportional share $V_i/n$. To determine whether their doubly logarithmic
dependence can be improved, define the realized approximation ratio against this benchmark,
after adding one outside good, as
$\rho_\mu(A)=\min(\{1\}\cup\{(v_i(A_i)+b_i)/\mu_i:\mu_i>0\})$.

\begin{restatable}{theorem}{thmfractionallower}
\label{thm:fractional-lower}
For $n\ge8$ and $n2^{-2n}\le\delta\le1/2$, the optimal worst-case approximation ratio against
the fractional benchmark, with probability at least $1-\delta$, is
$\Theta(1/\log\log(n/\delta))$. Moreover, for every fixed $\eps\in(0,1)$, guaranteeing
$v_i(A_i)+b_i^{(k)}\ge(1-\eps)\mu_i$ for every agent on every input with probability at least
$1-\delta$ requires $k=\Omega_\eps(\log\log(n/\delta))$. Similarly, any guarantee of the form
$v_i(A_i)\ge(1-\eps)\mu_i-C\nu_i$ for every agent on every input with probability at least
$1-\delta$ requires $C=\Omega_\eps(\log\log(n/\delta))$.
\end{restatable}

The optimal worst-case expected ratio $\inf_I\E[\rho_\mu(\calA,I)]$, optimized over randomized
online algorithms $\calA$, is also $\Theta(1/\log\log n)$. The lower bounds adapt a construction
of \citet[Theorem~4]{gaoEtAl2021improved}: goods are valued by pairs of agents, and later goods
reveal which member of each pair needed an earlier good (\cref{app:fractional-lower}). Notably,
every allocation of these hard instances is PROP1. Thus, the theorem establishes optimality
against the stronger fractional benchmark $\mu_i$ and does not imply that the
$\Theta(1/\log\log(n/\delta))$ dependence is optimal for PROP1.

\paragraph{Exact ex-ante proportionality.}
Our realized guarantees do not ensure exact ex-ante proportionality. In particular, $\DA$ may
favor an agent with a large accumulated deficit at one scale at the expense of another agent's
expected utility at a different scale. Indeed, $\DA$ fails ex-ante proportionality even with
only two agents and two goods (\cref{prop:not-ex-ante}). This limitation is specific to $\DA$:
whether an online algorithm can combine our realized guarantees with exact ex-ante
proportionality remains open.

\section{Lower Bounds for Other Fairness Notions}
\label{sec:lower-bounds}

The preceding results show that approximate proportionality can be achieved with probability
arbitrarily close to one. We now show a strong contrast for EF1, MMS, and EFX: for every
randomized online algorithm and every positive approximation factor, there is a fixed input on
which the approximation cannot be achieved with high probability. Thus, the separation from
proportionality concerns not only the approximation factor, but also the probability of achieving
any positive factor; the hard input may depend on the algorithm and the desired factor
$\alpha$. Our lower bounds use Yao's principle \citep{yao1977probabilistic}. For each result,
we construct a finite distribution over inputs fixed in advance and bound every deterministic
algorithm's success probability under it. Fixing a randomized algorithm's internal randomness
gives the same bound on a random input, so some fixed input satisfies it as well. None of our
constructions depends on the algorithm's realized allocations.

\subsection{Exponentially small success probability for EF1}
\label{subsec:ef1-lower-bound}

Although EF1 allocations always exist offline \citep{liptonEtAl2004approximately}, approximate
EF1 is much harder to achieve online: even against an oblivious adversary, randomization cannot
give any positive approximation with high probability in the worst case.

\begin{restatable}{theorem}{thmefoneconfidence}
\label{thm:ef1-confidence}
For all sufficiently large $n$, every $\alpha\in(0,1]$, and every randomized fully online
algorithm $\calA$, there is a fixed input $I$ with $O(n)$ goods such that
$\Prb[\rho_{\mathrm{EF1}}(\calA,I)\ge\alpha]\le e^{-\Omega(n)}$.
This holds even if every value is strictly positive, all values belong to four levels in
$(0,1]$, and the algorithm knows the number of goods in advance. The constant in the exponent
is independent of $\alpha$.
\end{restatable}

We outline the construction and the source of the lower bound. Fix a positive multiple $q$ of
$16$ with $n\ge1024q$, and let $r=8q$. Before the run, choose a uniformly random $r$-element
set $S$ of \emph{target agents}, assign each agent outside $S$ an independent uniform group
label $\ell_i\in[q]$, and choose independent fair bits $X_1,\dots,X_q$. For the construction,
first use the unscaled values $K=4/\alpha$, $\eta=\alpha/(4r)$, and
$\zeta=\alpha\eta/(4q)$. Present $n-r$ \emph{unit goods}, each worth one to every agent.
Next, for each target $a\in S$, present one good worth $K$ to $a$ and $\eta$ to every other
agent. Finally, present $q$ goods $f_1,\dots,f_q$. If $X_j=0$, then $f_j$ is worth $\zeta$
to every agent; if $X_j=1$, it is worth $K$ to every target, one to each nontarget agent with
label $j$, and $\zeta$ to everyone else. Dividing all values by $K$ places the four resulting
value levels in $(0,1]$ without revealing the hidden choices.

The first $n$ \emph{initial goods} reveal neither the group labels nor the bits. If many agents receive none, an unserved nontarget needs her group's final good to have high value and be allocated to her. Allocating initial goods to almost all nontargets creates the opposite difficulty: few valuable final goods can then be allocated without envy persisting after one good is removed. In particular, giving a final good to an agent who already holds a unit good creates such envy for any unserved member of that group. Thus, in either case, success requires many independent hidden choices to be favorable.

The full argument in \cref{app:ef1-proof} bounds the success probability by $4e^{-q/128}$ on
an input with exactly $n+q$ goods. Taking $q=\Theta(n)$ gives the theorem. More generally, for
any fixed $a>0$, only $O_a(\ln n)$ goods beyond the first $n$ suffice to obtain a success
probability of at most $n^{-a}$ for sufficiently large $n$. For fixed $\alpha$ and
$q=\Theta(n)$, the smallest value is $\Theta_\alpha(n^{-2})$, so the construction does not
require exponentially separated values.

\begin{restatable}{corollary}{corefoneexpected}
\label[corollary]{cor:ef1-expected}
Every randomized fully online algorithm $\calA$ satisfies
$R^{\calA}_{\mathrm{EF1}}(n)\le e^{-\Omega(n)}$.
\end{restatable}

For each $\alpha$, the inequality
$\E[\rho]\le\alpha+\Prb[\rho\ge\alpha]$, followed by taking the infimum over inputs and then
letting $\alpha$ tend to zero, gives the corollary (\cref{app:ef1-expected}). Together with
\cref{thm:main}, this yields an exponential separation between the achievable worst-case
expected realized ratios for PROP1 and EF1. Intuitively, a good that benefits one agent can
create substantial envy for another even when the latter's own proportional share is small.

\subsection{Success probability for MMS}
\label{subsec:mms-lower-bound}

For MMS, a constant upper bound on the success probability already separates it from PROP1,
which admits success probability arbitrarily close to one. The two-agent case follows from a
construction of \citet{zhouEtAl2023multiagent}; we extend it to every $n$.

\begin{restatable}{theorem}{thmmmsconfidence}
\label{thm:mms-confidence}
For every $n\ge2$, every $\alpha\in(0,1]$, and every randomized fully online algorithm
$\calA$, there is a fixed input $I$ such that
$\Prb[\rho_{\text{MMS}}(\calA,I)\ge\alpha]\le1-\lfloor n/2\rfloor/(2n)\le5/6$.
The input has exactly $2n-\lfloor n/2\rfloor$ goods, all values in $[0,1]$, and the number
of goods may be known in advance.
\end{restatable}

The bound equals $3/4$ for even $n$. The construction chooses a uniformly random target
agent, whose identity is not revealed by the initial goods, and one of two endings. If the
final goods have value zero, any positive MMS approximation requires every agent to receive
an earlier positively valued good. If they instead have large values, the same early allocation
leaves too many agents, including the target, competing for them. Thus, allocations that
succeed under one ending tend to fail under the other. The counting argument in
\cref{app:mms-proof} yields the probability bound above.

Importantly, the bound remains nontrivial even as $\alpha$ tends to zero, unlike an upper bound
derived only from the expected MMS ratio. The construction relies on goods that are large
relative to the agents' proportional shares; under the small-good condition of
\cref{cor:small-goods}, our positive result instead guarantees $(1-\eps)$-MMS with probability
at least $1-\delta$.

\subsection{EFX with only \texorpdfstring{$n+1$}{n+1} goods}
\label{subsec:efx-lower-bound}

With at most $n$ goods, assigning each good to a distinct agent yields exact EFX. With just one
additional good, however, the worst-case success probability drops from one to at most $1/n!$
for any positive approximation factor.

\begin{restatable}{theorem}{thmefxconfidence}
\label{thm:efx-confidence}
For every $n\ge2$, every $\alpha\in(0,1]$, and every randomized fully online algorithm
$\calA$, there is a fixed input $I$ with exactly $n+1$ goods such that
$\Prb[\rho_{\text{EFX}}(\calA,I)\ge\alpha]\le1/n!$.
All values are in $(0,1]$, and the number of goods may be known in advance.
Conversely, there is a randomized fully online algorithm that returns exact EFX with probability
at least $2/((n+1)n!)$ on every input with $n+1$ goods.
\end{restatable}

The construction chooses a uniformly random ordering $(\pi_1,\dots,\pi_n)$ of the agents.
For each agent, values increase rapidly until her position in this ordering and become small
thereafter. This makes every successful allocation assign $g_t$ to $\pi_t$ for $t<n$, while
the identity of $\pi_t$ is revealed only when $g_{t+1}$ arrives. The algorithm must therefore
guess correctly with conditional probabilities at most $1/n,1/(n-1),\dots,1/2$, whose product
is $1/n!$. For the converse, choose uniformly among allocations with one two-good bundle and
all other bundles singletons. By the counting argument of \citet{neohTeh2025efx}, at least $n$
of the $\binom{n+1}{2}n!$ such allocations are EFX. The valuations and full proof appear in
\cref{app:efx-proof}.

The same bounds hold for the optimal worst-case expected EFX ratio restricted
to $n+1$ goods. In contrast, exact PROP1 is achievable deterministically with this many goods:
assign the first $n$ goods to distinct agents and the last arbitrarily. Every agent then has at
most $n$ outside goods, so $V_i\le v_i(A_i)+nb_i\le n(v_i(A_i)+b_i)$.

Knowing the agents' maximum values in advance would reveal the ordering in this construction. We next show that the EFX success probability remains exponentially small even with this information.

\begin{restatable}{theorem}{thmefxknownmax}
\label{thm:efx-known-max}
For every $n\ge2$, every $\alpha\in(0,1]$, and every randomized online algorithm $\calA$,
there is a fixed input $I$ with exactly $n+1$ goods such that
\begin{equation}
 \Prb[\rho_{\text{EFX}}(\calA,I)\ge\alpha]
 \le\frac{1+\lfloor n^2/4\rfloor}{\binom n{\lfloor n/2\rfloor}}
 =O(n^{5/2}2^{-n}).
 \label{eq:efx-known-max}
\end{equation}
All values belong to $\{\alpha^2/16,\alpha/4,1\}$. The algorithm may know the number of
goods and every agent's exact maximum value $\nu_i=1$.
\end{restatable}

The second construction instead chooses a uniformly random set of agents that must receive
almost all of the early goods, before the set can be identified from the observed values. A
final good worth one to every agent makes all maximum values equal, but can compensate at most
one member of this set who received no early good (\cref{app:efx-known-max}). For exact EFX,
the three values are simply $\{1/16,1/4,1\}$, independently of $n$.

\section{Approximate Envy-Freeness up to \texorpdfstring{$k$}{k} Goods}
\label{sec:efk}

The EF1 lower bound raises the question of whether approximate envy-freeness becomes attainable
when more goods may be removed. We show that it does: logarithmically many removals in
$n/\delta$ suffice, even under independent allocation. Recall that the $\Like$ rule assigns
each good independently and uniformly among the agents in $N_t=\{i\in[n]:v_i(g_t)>0\}$, and assigns all-zero goods
arbitrarily \citep{aleksandrovEtAl2015online,neohTeh2026closing}. For comparison, $\Rand$
assigns each good uniformly among all $n$ agents. Only $\Like$ is non-wasteful.

\begin{restatable}{theorem}{thmefk}
\label{thm:efk}
For every fixed input $I$, every $\eps\in(0,1)$, and every positive integer $k$,
\begin{equation}
 \Prb[\rho_{\text{EF}k}(\Like,I)<1-\eps]
 \le n(n-1)e^{-\eps k/2}.
 \label{eq:efk-probability}
\end{equation}
Thus, $\Like$ returns a $(1-\eps)$-EF$k$ allocation with probability at least $1-\delta$
whenever $k\ge\lceil(2/\eps)\ln(n(n-1)/\delta)\rceil$. The same bounds hold for $\Rand$.
Neither rule requires $\eps$, $k$, $\delta$, or the number of goods as input.
\end{restatable}

The required number of removals is independent of the number and values of the goods. The key
idea is to compare, for each pair of agents, how many of one agent's most valuable goods are
allocated to herself and to the other agent. Since allocations are independent on a fixed
input, we may analyze the goods in decreasing order of value without changing the online rule.

\begin{proof}[Proof sketch of \cref{thm:efk}]
Fix $i\ne j$ and let $\alpha=1-\eps$. Order the goods valued positively by $i$ in decreasing value, breaking ties by arrival index, with values $a_1\ge\cdots\ge a_s>0$. Let $N_i(r)$ and $N_j(r)$ count assignments to $i$ and $j$ among the first $r$ goods, and set $S_r=N_j(r)-N_i(r)/\alpha$ and $z=\eps/2$. Let $p_i$ and $p_j$ be each good's assignment probabilities to $i$ and $j$. Since $i$ values it positively, $p_i\ge p_j$, and hence
\[
 \E[e^{z(S_r-S_{r-1})}\mid\text{earlier recipients}]
 \le1+p_j(e^z+e^{-z/\alpha}-2)\le1,
\]
where $e^z+e^{-z/\alpha}\le(1-z)^{-1}+\alpha/(\alpha+z)=2$.
Stopping when $S_r$ first exceeds $k$ shows that this occurs with probability at most
$e^{-\eps k/2}$. Otherwise,
$N_i(r)\ge\alpha\max\{0,N_j(r)-k\}$ for every $r$.

Remove from $A_j$ the $k$ goods most valuable to $i$, or all of them if fewer than $k$ are
present. Among the first $r$ positions, the number remaining in $A_j$ is
$\max\{0,N_j(r)-k\}$. Multiplying the resulting count inequalities by
$a_r-a_{r+1}\ge0$, with $a_{s+1}=0$, and summing yields
$v_i(A_i)\ge\alpha(v_i(A_j)-b_{ij}^{(k)})$. A union bound over the $n(n-1)$ ordered pairs
proves the theorem. The case $s=0$ is immediate.
\end{proof}

The complete stopping and summation argument appears in \cref{app:efk-proof}. Notably, no
union bound over goods is needed. Moreover, $p_i\ge p_j$ implies
$\E[v_i(A_i)]=\mu_i\ge\E[v_i(A_j)]$, so $\Like$ combines the realized EF$k$ guarantee with
exact ex-ante envy-freeness.

\subsection{Matching dependence for \texorpdfstring{$\Like$}{Like}}
\label{subsec:efk-like-lower}

For $\Like$, the logarithmic dependence of $k$ on $n/\delta$ is necessary. Pair the agents and give each pair goods worth one to its two members and zero to everyone else. Independent allocation can leave one member with substantially fewer goods, while approximate EF$k$ must hold simultaneously across all pairs.

\begin{restatable}{proposition}{proplikeefklower}
\label[proposition]{prop:like-efk-lower}
For every $n\ge2$, $0<\eps\le1/4$, and positive integer $k$ with $\eps k\ge4$, there is a
fixed binary-valued input $I$ with $\kappa=2$ such that
\begin{equation}
 \Prb[\rho_{\text{EF}k}(\Like,I)\ge1-\eps]
 \le\exp\left(-\tfrac12\lfloor n/2\rfloor e^{-13\eps k}\right).
 \label{eq:efk-like-lower}
\end{equation}
Consequently, in this range, guaranteeing success probability at least $1-\delta$ on every
input, for $0<\delta\le1/2$, requires
$k\ge(13\eps)^{-1}\ln(\lfloor n/2\rfloor/(4\delta))$ whenever the right-hand side is positive.
\end{restatable}

Give each pair $2\lceil k/\eps\rceil$ goods. A binomial estimate shows that the probability of violating the within-pair comparison is at least $\tfrac12e^{-13\eps k}$; independence across pairs then gives the result (\cref{app:efk-like-lower}). Thus, in the stated range, the upper and lower bounds have the same dependence on $n$, $\delta$, and $\eps$. This lower bound is specific to $\Like$: Balanced $\Like$ obtains exact EF1 on binary valuations \citep{aleksandrovEtAl2015online}.

\begin{restatable}{corollary}{corefkexpected}
\label[corollary]{cor:efk-expected}
For every $n\ge2$ and positive integer $k$,
$R^{\Like}_{\text{EF}k}(n)\ge\max\{0,1-2(\ln(n(n-1))+1)/k\}$, and the same lower bound holds
for $\Rand$. For all sufficiently large $n$ and every integer $k\ge\ln n$,
$1-R^{\Like}_{\text{EF}k}(n)=\Theta(\ln n/k)$.
\end{restatable}

Integrating \eqref{eq:efk-probability} gives the first bound; applying \cref{prop:like-efk-lower} with $\eps=\ln n/(26k)$ gives the matching asymptotic bound (\cref{app:efk-expected}). Whether using allocation history, as $\DA$ does for proportionality, can reduce the required removals remains open.

\section{Conclusion and Future Work}
\label{sec:conclusion}

We have shown that allocation history and value magnitudes can substantially improve realized
fairness against an oblivious adversary. Our Deficit Allocation algorithm improves the high-probability PROP1
approximation from inverse-logarithmic to inverse-doubly-logarithmic in $n/\delta$. Its refined
variant gives near-proportionality when individual goods are sufficiently small and near-optimal
PROP$k$ with doubly logarithmically many outside goods. For the stronger fractional benchmark,
these dependences are asymptotically tight over a broad range of $\delta$, while algorithms
ignoring value magnitudes cannot improve the logarithmic dependence.

Other notions behave differently. In particular, any positive EF1 approximation can have exponentially small
success probability, EFX can have probability at most $1/n!$ with only $n+1$ goods, and the
probability of any positive MMS approximation can remain bounded away from one. For EF$k$,
$\Like$ achieves near-optimal fairness with high probability for logarithmic $k$, and this
dependence is necessary for $\Like$.

The main open question is whether a constant worst-case expected realized PROP1 guarantee is
possible on unrestricted inputs, or whether some dependence on $n$ is necessary. It is also
open whether our doubly logarithmic realized guarantee can coexist with exact ex-ante
proportionality, and whether allocation history can reduce the logarithmic number of removals
required for approximate EF$k$.

\bibliographystyle{plainnat}
\bibliography{bib}

\clearpage
\appendix
\numberwithin{equation}{section}

\section{Further Related Work}
\label[appendix]{app:further-related-work}

This appendix complements \cref{subsec:related-work}. We discuss offline allocation with both
ex-ante and realized guarantees, unequal entitlements, restricted valuations and stochastic
arrivals, online models with reassignments or subsidies, online welfare maximization and maximin
shares, the probabilistic techniques most closely related to our deficit analysis, and other
sequential models in social choice.

\paragraph{Ex-ante and realized fairness offline.}
The tension between ex-ante and realized fairness is also central to offline ``best of both
worlds'' allocation. \citet{azizEtAl2024best} give an allocation that is ex-ante envy-free and
EF1 in every realization, and \citet{babaioffEtAl2022best} obtain ex-ante proportionality
together with PROP1 and $1/2$-MMS in every realization. Offline randomized rounding can also
preserve fractional utilities in expectation while limiting each agent's realized loss to the
value of one good \citep{budishEtAl2013designing,babaioffEtAl2022best}. These algorithms,
however, use all valuations before allocating any good. Our results quantify the simultaneous
realized guarantees that remain attainable when each good must be allocated before future goods
are known.

\paragraph{Unequal entitlements.}
Fair division among agents with unequal entitlements has been studied extensively offline; see
\citet{suksompong2025weighted} for a review. For instance, \citet{chakrabortyEtAl2021weighted}
introduce weighted envy-freeness up to one good, and \citet{farhadiEtAl2019fair} introduce the
weighted maximin share. Weighted envy-freeness has since been extended to submodular valuations
\citep{montanariEtAl2025weighted} and to mixed manna \citep{teh2026weighted}, and weighted rules
have been studied for binary valuations, through maximum weighted Nash welfare
\citep{suksompongTeh2022binary}, and for matroid-rank valuations, with respect to monotonicity and
strategyproofness \citep{suksompongTeh2023matroid}. Weighted proportionality up to one good can
be attained offline together with Pareto optimality \citep{azizEtAl2020polynomial};
\cref{cor:weighted} gives a realized online guarantee for this notion, and more generally for
weighted PROP$k$, against an oblivious adversary. It divides each good in proportion to the
entitlements of the agents who value it positively, and its constants do not depend on the
entitlements.

\paragraph{Restricted valuations and stochastic arrivals.}
Additional structure permits stronger envy guarantees. \citet{aleksandrovEtAl2015online}
introduce $\Like$ and Balanced $\Like$; for binary valuations, Balanced $\Like$ achieves EF1 in
every realization. Offline, binary valuations are especially well understood: maximum Nash welfare
with lexicographic tie-breaking is EF1, Pareto optimal, and group strategyproof
\citep{halpernEtAl2020binary}, \citet{babaioffEtAl2021dichotomous} obtain truthful mechanisms
that are EFX and maximize Nash welfare for submodular valuations with binary marginals, and
\citet{brandlEtAl2026binary} characterize maximum Nash welfare axiomatically in this setting.
Online, \citet{wangWei2026online} give EF1 and MMS guarantees for restricted
valuation classes, including exact guarantees for binary valuations, and rule out positive
deterministic approximations with three positive values. Under independent draws from a known
distribution over a fixed finite set of valuation vectors, \citet{benadeEtAl2024fair} obtain
ex-post Pareto efficiency together with pairwise EF1 or envy-freeness with high probability as
the number of goods grows, and \citet{gaoEtAl2021market} obtain asymptotic proportionality and
envy-freeness under independent draws from a fixed distribution, without knowing that
distribution. Our guarantees, in contrast, apply to every fixed input with unrestricted
nonnegative additive valuations, and our lower bounds quantify the success probability of
randomized algorithms on inputs fixed in advance. \citet{cohen2026budgetconstraints} study
online allocation under budget constraints, where deterministic guarantees for feasible
envy-freeness require bounded value-to-size density spread; offline, budget-feasible fair
allocation has been studied for both goods \citep{wuEtAl2021budget} and chores
\citep{elkindEtAl2024budget}.

\paragraph{Reassignments and subsidies.}
Other online models relax irrevocability or allow monetary transfers. \citet{huangTeh2026upgrade}
study MMS when goods may be reassigned in ways that leave every affected agent no worse off:
if each held good may only be swapped for one that the agent values at least as much, the optimal
MMS factor is about $1/\ln n$, whereas constant factors are attainable if only the value of each
agent's bundle must not decrease. Our allocations are irrevocable, and \cref{thm:mms-confidence}
bounds the success probability of every randomized algorithm for any positive MMS factor.
Offline, envy can be eliminated with bounded subsidies
\citep{halpernShah2019subsidy,brustleEtAl2020one}. \citet{tehEtAl2026house} study an online
version of house allocation with subsidies, in which each agent holds a single house: envy-freeability
can be maintained with bounded recourse, but the total subsidy admits no bounded competitive
ratio for randomized algorithms, even against a non-adaptive adversary.

\paragraph{Welfare, maximin shares, and competitive analysis.}
Another line of work studies online welfare maximization. \citet{cohenAgmon2024near} consider
indivisible goods arriving in uniformly random order, and \citet{banerjeeEtAl2022predictions}
study Nash social welfare for divisible goods using predictions of total values. Such
guarantees compare with an optimal offline allocation, whereas our benchmarks concern each
agent's fair share and every pairwise envy comparison. \citet{chenTan2026competitive} study
online fairness through competitive ratios relative to the best offline fairness factor. Even
offline, exact MMS allocations need not exist, although constant-factor approximations always
do \citep{kurokawaEtAl2018fair}; our MMS results therefore concern approximation factors. Online,
\citet{zhouEtAl2023multiagent} study deterministic MMS allocation and give a two-agent
impossibility without advance knowledge of total values. Our two-agent bound on the success
probability follows by choosing one of the inputs in their construction uniformly at random
before the run; the new part is the bound for arbitrary $n$. With identical valuations, expected
realized MMS is exactly the expected competitive factor for online machine covering, so an
$O(n^{-1/2})$ upper bound already follows from \citet{azarEpstein1998covering}. Our bound on
the success probability instead applies even when the positive approximation factor is
arbitrarily small, and our small-good result complements these bounds by giving
near-proportionality when each good is small relative to the proportional share.

\paragraph{Correlated selection, balanced allocations, and scheduling.}
Our selection analysis is related to online correlated selection, introduced for online matching
by \citet{fahrbachEtAl2022matching} and extended to selection among more than two elements by
\citet{gaoEtAl2021improved} and \citet{blancCharikar2021multiway}. These works bound the
probability that an element is not selected during specified rounds. We instead bound the
deficit that remains after repeated selections: each received good accounts for a bounded
fractional amount, while the remainder continues to affect later allocations. The induction and
concavity argument for selection outside a set follows \citet[Theorem~9]{gaoEtAl2021improved};
we additionally account for selection within the set, since a selected element may retain a
positive deficit (\cref{app:deficits}). Our lower bound for the fractional benchmark
(\cref{thm:fractional-lower}) adapts the construction that \citet[Theorem~4]{gaoEtAl2021improved}
use to show that their semi-online correlated selection is optimal; the fair-division
instantiation, the use of parallel copies to calibrate the failure probability, and the
consequences for the fractional benchmark are ours.

The doubly exponential tail of our deficits is reminiscent of the power of two choices in
balls-into-bins processes, where placing each ball in the less loaded of two randomly chosen
bins reduces the maximum load from nearly logarithmic to doubly logarithmic
\citep{azarEtAl1999balanced}. \citet{berenbrinkEtAl2006balanced} show that this doubly
logarithmic gap above the average load persists however many balls are allocated, much as our
additive loss does not depend on the number of goods. The analogy is only qualitative: there,
the candidates for each ball are sampled at random, whereas the agents who value each good are
fixed by the adversary, and fairness requires a guarantee for every agent at every scale.
Related probability bounds also arise in processor scheduling. \citet{benderEtAl2019backlog}
obtain doubly exponential probability bounds in a model where fractional amounts are added each
round and one accumulated total may then be reduced. Their algorithm may choose any total to
reduce, whereas ours can reduce only the recipient's deficit at the scale of the current good,
so their result does not directly imply our guarantee.

\paragraph{Other sequential models in social choice.}
In temporal voting, a fixed group of voters makes a collective decision in each of several rounds
\citep{freemanEtAl2017dynamic,aloufHeffetzEtAl2022better,lackner2020perpetual,bulteauEtAl2021justified,lacknerMaly2023proportional,chandakEtAl2024proportional,neohTeh2024welfare,elkindEtAl2025verifying,elkindEtAl2025chores,neohTeh2025strategic,phillipsEtAl2026strengthening,teh2026price,teh:hal-05759723,hazonEtAl2026ranked,zechEtAl2024aversion};
see \citet{elkindEtAl2024multiwinner} for a unified framework. Typically, one candidate is
selected in each round, and every voter who approves it benefits. Each decision thus affects all
voters, as in fair public decision making \citep{conitzerEtAl2017public}, whereas in our model
each good benefits only the agent who receives it. Fairness in temporal voting is usually
captured by proportional representation axioms from multiwinner voting, such as justified
representation, which concern groups of voters with common approvals rather than individual
agents. Some works also consider proportionality for individual agents, which is closer to
proportionality in fair division \citep{elkindEtAl2022slot,elkindEtAl2024elections}.

Closer to fair division, several works allocate the same items to the same agents over multiple
rounds. \citet{igarashiEtAl2024repeated} study when envy-freeness or proportionality over all
rounds is compatible with Pareto optimality. When the items are matched to the agents in each
round, \citet{caragiannisNarang2024repeatedly} study EF1 over all rounds, where an agent's value
for an item may depend on how many times she has received it before;
\citet{micheelWilczynski2024repeated} study envy-based fairness over time with ordinal
preferences; and \citet{limEtAl2026repeated} maximize the smallest total utility of an agent.
\citet{cooksonEtAl2025temporal} allow a different set of goods on each day and seek allocations
that are EF1 not only on each day but also up to each day. These works assume that the
preferences for all rounds are known in advance, whereas in our model the values for each good
are revealed only when it arrives. Repeated matching has also been studied with preferences that
change over time, both with ordinal preferences \citep{hosseiniEtAl2015matching} and in
two-sided markets \citep{gollapudiEtAl2020almost}; in the latter, EF1 is achievable for symmetric
binary valuations without using future valuations.

Related work also studies the online formation of coalitions and clusters~\citep{cohen2024onlinefriends,cohen2024online,cohen2025online,cohen2025decentralized,cohen2023online,cohen2025fair,cohen2026stability,cohen2026delayed}.
These models involve assigning arriving agents or points to groups, or learning coalition structures through repeated interactions. Their objectives include social welfare, stability, regret minimization, and clustering costs. In contrast, our agents are fixed, indivisible goods arrive sequentially, and each must be allocated immediately and irrevocably. Fairness is evaluated through agents' valuations of allocated goods, rather than their preferences over coalition members or the quality of a clustering.

\section{Omitted Proofs in Section~\ref{sec:main-result}}
\label[appendix]{app:sec-main-result}

\subsection{Joint Deficit Lemma}
\label[appendix]{app:deficits}

This subsection proves the probability bound used in \cref{thm:main} and its extension for
\cref{thm:near-proportional}. We state the argument for a finite set $\calE$ of elements. In
$\DA$, an element is an agent--scale pair; in the refined algorithm, it is one of the copies
of such a pair. The elements and their fractional amounts are fixed by the input, even though
the algorithm need not know them before they occur. The induction and concavity argument
for selection outside a set follows \citet[Theorem~9]{gaoEtAl2021improved}. The additional
argument handles selection within the set, since a selected element can retain a positive
deficit.

In round $t$, a vector $(x_e^t)_{e\in\calE}$ of nonnegative fractional amounts is revealed,
with $\sum_e x_e^t=1$. The sequence of these vectors is fixed in advance. Start with
$D_e^0=0$ for every $e$, and select an element $J_t$ with probability proportional to
$x_e^t w(D_e^{t-1})$. Then, set $D_e^t=\max\{0,D_e^{t-1}+x_e^t-c\one\{J_t=e\}\}$ for every
$e\in\calE$. Only an element with a positive fractional amount can be selected. The positive constant $c$
is the decrease applied on selection.

\begin{lemma}
\label[lemma]{lem:joint-deficit}
Suppose $0\le x_e^t\le\xi\le1$ in every round, where $\xi>0$, and let $c>\xi$.
For parameters $\gamma>0$, $\beta>1$, and $\theta\in(0,1)$, let $H(y)=\gamma(\beta^y-1)$,
$w(y)=e^{\theta H(y)}$, $C=1/(1-\theta)$, and $\tau=\exp(-(1-\theta)\gamma(\beta^{c-\xi}-1))$.
Assume that
\begin{equation}
 \begin{gathered}
 \beta-1\le\theta,\quad (1-\theta)\beta^{c-\xi}\ge1,\\
 \gamma(\beta-1)\le1-\ln C-C\tau.
 \end{gathered}
 \label{eq:deficit-conditions}
\end{equation}
Then, for every round $t\ge0$, every $S\subseteq\calE$, and all thresholds $r_e\ge0$,
\begin{equation}
 \Prb[D_e^t\ge r_e\text{ for all }e\in S]
 \le\exp\left(-\sum_{e\in S}H(r_e)\right).
 \label{eq:joint-deficit}
\end{equation}
\end{lemma}

The proof considers separately whether the selected element belongs to $S$. If it does not,
a concavity argument controls the weights outside $S$. If it does, its previous deficit must
exceed the threshold needed without selection by at least $c-\xi$. The resulting increase in
the deficit makes this second event sufficiently unlikely to preserve \eqref{eq:joint-deficit}.
Choosing $w=e^{\theta H}$ with $\theta<1$ also keeps the relevant weighted expectations finite.

\begin{proof}
We induct on $t$, with the induction hypothesis covering every set and every threshold vector.
The claim is immediate at $t=0$. Suppose it holds after round $t-1$, and fix $S$ and the
thresholds at round $t$. Thresholds equal to zero may be omitted, so assume all $r_e>0$.
Let $x_e=x_e^t$ and $D_e=D_e^{t-1}$, and set
\[
 \begin{aligned}
 a_e&=\max\{0,r_e-x_e\}, & q&=\sum_{e\in S}x_e,\\
 z&=\sum_{e\in S}x_e w(a_e), &
 P&=\exp\left(-\sum_{e\in S}H(a_e)\right).
 \end{aligned}
\]
These quantities are deterministic. If $q=0$, no element in $S$ changes, and the claim follows
directly from induction. We may therefore assume that $q>0$, so $z>0$.

Let $X=\one\{D_e\ge a_e\text{ for every }e\in S\}$ and $Y=X\sum_{j\notin S}x_jw(D_j)$.
The event in \eqref{eq:joint-deficit} requires $X=1$. The induction hypothesis gives
$\E X\le P$. For any $j\notin S$, applying it to $S\cup\{j\}$ and integrating yields
\begin{align}
 \E[Xw(D_j)]
 &=\Prb[X=1]+\int_0^\infty
    w'(u)\Prb[X=1,D_j>u]\,du\notag\\
 &\le P\left(1+\int_0^\infty w'(u)e^{-H(u)}\,du\right)
 =CP.
 \label{eq:weighted-joint-tail}
\end{align}
The last equality holds because $w'(u)=\theta H'(u)e^{\theta H(u)}$, so the integral equals
\[
 \theta\int_0^\infty H'(u)e^{-(1-\theta)H(u)}\,du=\frac{\theta}{1-\theta}.
\]
Consequently,
$\E Y\le C(1-q)P$, since $\sum_{j\notin S}x_j=1-q$.

\paragraph{Selection outside $S$.}
On $X=1$, we have $\sum_{e\in S}x_ew(D_e)\ge z$ because $w$ is increasing. Hence, given the
deficits before round $t$, the probability that the event in \eqref{eq:joint-deficit} occurs
and $J_t\notin S$ is at most $Y/(z+Y)$ when $X=1$, and zero when $X=0$. In both cases, it is
at most $f_z(X,Y)$, where
$f_z(a,b)=ab/(za+b)$ for $(a,b)\ne(0,0)$ and $f_z(0,0)=0$. The function is nondecreasing in each coordinate and jointly concave on the nonnegative
quadrant. Indeed, its partial derivatives are $b^2/(za+b)^2$ and $za^2/(za+b)^2$, and its
Hessian away from the origin is
\[
 \frac{2z}{(za+b)^3}
 \begin{pmatrix}-b^2&ab\\ab&-a^2\end{pmatrix},
\]
which is negative semidefinite. Continuity at the origin follows from
$0\le f_z(a,b)\le\min\{a,b/z\}$. Jensen's inequality, monotonicity, and the bounds
$\E X\le P$ and $\E Y\le C(1-q)P$ therefore give
\begin{equation}
 \Prb[J_t\notin S\text{ and }D_e^t\ge r_e\text{ for all }e\in S]
 \le P\frac{C(1-q)}{z+C(1-q)}.
 \label{eq:selection-outside}
\end{equation}

\paragraph{Selection in $S$.}
Fix $e\in S$. If $J_t=e$ and $D_e^t\ge r_e$, then, because $r_e>0$, we have
$D_e\ge r_e+c-x_e\ge a_e+c-\xi$. Let $\calE_e$ be the event that $D_e\ge a_e+c-\xi$ and $D_j\ge a_j$ for all
$j\in S\setminus\{e\}$. This event is determined before $J_t$ is selected. On $\calE_e$,
the denominator of the selection probabilities is at least $z+1-q$, since $w$ is increasing
and $w\ge1$. Hence, given the deficits before round $t$, the probability that $J_t=e$ and
$D_j^t\ge r_j$ for all $j\in S$ is at most $x_ew(D_e)\one\{\calE_e\}/(z+1-q)$.

To bound the expectation of $w(D_e)\one\{\calE_e\}$, consider any $y\ge0$ and write
$w(D_e)\one\{D_e\ge y\}=w(y)\one\{D_e\ge y\}+\int_y^\infty w'(u)\one\{D_e>u\}\,du$. Applying the
induction hypothesis to each resulting probability, the first term contributes at most
$\exp(-\sum_{j\in S\setminus\{e\}}H(a_j)-(1-\theta)H(y))$, and, as in
\eqref{eq:weighted-joint-tail}, the integral contributes at most $C-1=\theta/(1-\theta)$ times
this amount. Taking $y=a_e+c-\xi$ therefore gives
\begin{equation}
 \begin{aligned}
 \E[w(D_e)\one\{\calE_e\}]
 &\le C\exp\left(-\sum_{j\in S\setminus\{e\}}H(a_j)
                      -(1-\theta)H(a_e+c-\xi)\right)\\
 &=CP\exp(H(a_e)-(1-\theta)H(a_e+c-\xi))\le CP\tau.
 \end{aligned}
 \label{eq:selected-truncated-tail}
\end{equation}
For the last inequality in \eqref{eq:selected-truncated-tail}, the identity
$H(a+c-\xi)=\beta^{c-\xi}H(a)+\gamma(\beta^{c-\xi}-1)$, together with $H(a_e)\ge0$ and the
second condition in \eqref{eq:deficit-conditions}, gives
$H(a_e)-(1-\theta)H(a_e+c-\xi)\le-(1-\theta)\gamma(\beta^{c-\xi}-1)=\ln\tau$.

The probability that $J_t=e$ and $D_j^t\ge r_j$ for all $j\in S$ is therefore at most
$x_eCP\tau/(z+1-q)$. Summing over $e\in S$
and using $z+C(1-q)\le C(z+1-q)$ gives
\begin{equation}
 \Prb[J_t\in S\text{ and }D_e^t\ge r_e\text{ for all }e\in S]
 \le P\frac{C^2\tau q}{z+C(1-q)}.
 \label{eq:selection-inside}
\end{equation}
Combining \eqref{eq:selection-outside} and \eqref{eq:selection-inside}, we obtain
\begin{equation}
 \Prb[D_e^t\ge r_e\text{ for every }e\in S]
 \le P\frac{1-q+C\tau q}{1-q+z/C}.
 \label{eq:one-round-ratio}
\end{equation}

\paragraph{Completing the induction.}
The weighted arithmetic--geometric mean inequality for the numbers $1$ and $w(a_e)/C$
($e\in S$), with weights $1-q$ and $x_e$, respectively, gives
$\ln(1-q+z/C)\ge\sum_{e\in S}x_e(\theta H(a_e)-\ln C)$. The last condition in \eqref{eq:deficit-conditions} implies $C\tau<1$, so
$-\ln(1-q+C\tau q)\ge(1-C\tau)q$. Hence
\begin{align*}
 \ln\frac{1-q+z/C}{1-q+C\tau q}
 &\ge\theta\sum_{e\in S}x_eH(a_e)+q(1-\ln C-C\tau)\\
 &\ge(\beta-1)\sum_{e\in S}x_e(H(a_e)+\gamma).
\end{align*}
Since $0\le x_e\le1$, convexity gives $\beta^{x_e}-1\le(\beta-1)x_e$. Also
$r_e\le a_e+x_e$. Therefore
\[
 H(r_e)-H(a_e)
 \le H(a_e+x_e)-H(a_e)
 \le(\beta-1)x_e(H(a_e)+\gamma).
\]
Summing over $e\in S$ and comparing with the previous display, we find that the fraction in
\eqref{eq:one-round-ratio} is at most
\[
 \exp\left(-\sum_{e\in S}\bigl(H(r_e)-H(a_e)\bigr)\right).
\]
Since $P=\exp(-\sum_{e\in S}H(a_e))$, this proves \eqref{eq:joint-deficit} and completes the
induction.
\end{proof}

\subsection{Applying the Joint Deficit Lemma to Deficit Allocation}
\label[appendix]{app:simple-parameters}

For $\DA$, apply \cref{lem:joint-deficit} with $\gamma=1$, $\beta=4/3$, $\theta=1/3$, $c=6$,
and $\xi=1$. The first two conditions in \eqref{eq:deficit-conditions} are immediate. For the
third, $C=3/2$ and $\tau=\exp(-(2/3)((4/3)^5-1))<e^{-2}<1/7$. Using $\ln(3/2)<5/12$, we get
$1-\ln C-C\tau>1-5/12-3/14=31/84>1/3=\gamma(\beta-1)$. Thus all three conditions hold. After omitting all-zero arrivals, the fractional amounts of
$\DA$ sum to one in each round. They are fixed by the input, so \cref{lem:joint-deficit}
applies to the finite set of encountered agent--scale pairs. Its singleton case gives the
marginal tail bound used in \cref{app:main-proof}.

\subsection{Proof of Theorem~\ref{thm:main}}
\label[appendix]{app:main-proof}

\thmmain*

\begin{proof}
Let $B_{n,\delta}$ be as in \cref{thm:main}, so that $H(B_{n,\delta})=2\ln(2n/\delta)$,
and let $A$ be the allocation returned by $\DA$. We first show that, with probability at least
$1-\delta$, every agent $i$ satisfies
\begin{equation}
 \mu_i\le12v_i(A_i)+4B_{n,\delta}\nu_i,
 \label{eq:additive}
\end{equation}
and then derive the PROP1 guarantee from \eqref{eq:additive}. All sums over scales below range
over the scales encountered on the fixed input, and $D_{i,d}$ denotes the final deficit.

For each agent $i$ and scale $d$, let $Y_{i,d}$ be the total fractional amount assigned to that
pair and $N_{i,d}$ the number of goods allocated to $i$ at that scale. Summing the deficit
updates, and observing that truncation at zero can only increase the resulting deficit, gives
$Y_{i,d}\le6N_{i,d}+D_{i,d}$. Every good at scale $d$ is worth at least $2^d$ and less than
$2^{d+1}$ to agent $i$. Hence
\begin{equation}
 \begin{aligned}
 \mu_i&\le2\sum_d2^dY_{i,d}\\
 &\le12\sum_d2^dN_{i,d}+2\sum_d2^dD_{i,d}\\
 &\le12v_i(A_i)+2\sum_d2^dD_{i,d}.
 \end{aligned}
 \label{eq:deficit-accounting}
\end{equation}
The coefficient of $v_i(A_i)$ is constant: only the final deficits need a probabilistic bound.

\Cref{lem:joint-deficit} in \cref{app:deficits}, applied with the parameters in
\cref{app:simple-parameters} and $S=\{(i,d)\}$, shows that, on every input, the final deficits
satisfy $\Prb[D_{i,d}\ge y]\le e^{-H(y)}$ for all $y\ge0$.

To combine the scales, define the convex and increasing function $\phi(y)=\exp(H(y)/2)$.
Integrating this tail bound gives
\begin{equation}
 \E[\phi(D_{i,d})]
 \le1+\int_0^\infty\phi'(y)e^{-H(y)}\,dy=2.
 \label{eq:deficit-moment}
\end{equation}
Fix an agent with $\nu_i>0$. Every scale $d$ encountered for agent $i$ satisfies $2^d\le\nu_i$,
so the geometric sum satisfies $\sum_d2^d\le2\nu_i$. Apply convexity with weights
$2^d/(2\nu_i)$, placing any remaining weight at zero. Since $\phi(0)=1$,
\eqref{eq:deficit-moment} yields $\E[\phi(\sum_d2^dD_{i,d}/(2\nu_i))]\le2$. Markov's
inequality and $H(B_{n,\delta})=2\ln(2n/\delta)$ now imply
\begin{equation}
 \Prb\left[\sum_d2^dD_{i,d}>2B_{n,\delta}\nu_i\right]
 \le2e^{-H(B_{n,\delta})/2}=\frac{\delta}{n}.
 \label{eq:weighted-deficit-tail}
\end{equation}
An agent with $\nu_i=0$ has $\mu_i=0$, so she satisfies \eqref{eq:additive}
deterministically. A union bound over agents, followed by \eqref{eq:deficit-accounting}, proves
that, with probability at least $1-\delta$, \eqref{eq:additive} holds for every agent.

It remains to derive the PROP1 guarantee. A most valuable good for agent $i$ either belongs to
her bundle or lies outside it, so $\nu_i\le v_i(A_i)+b_i$. Hence, whenever \eqref{eq:additive} holds
for every agent, each agent satisfies $\mu_i\le(12+4B_{n,\delta})(v_i(A_i)+b_i)$. Since
$\mu_i\ge V_i/\kappa$, every agent then satisfies
$v_i(A_i)+b_i\ge n/(\kappa(12+4B_{n,\delta}))\cdot V_i/n$, so the allocation is
$\min\{1,n/(\kappa(12+4B_{n,\delta}))\}$-PROP1. Finally, $2n/\delta>4$ implies
$B_{n,\delta}>4$, so $12+4B_{n,\delta}<7B_{n,\delta}$. Moreover, $n/\delta>2$ implies
$1+2\ln(2n/\delta)<e^2\ln(n/\delta)$, so
$B_{n,\delta}<(2+\ln\ln(n/\delta))/\ln(4/3)\le11\log\log(n/\delta)$. Therefore the allocation is
$\Omega(\min\{1,n/(\kappa\log\log(n/\delta))\})$-PROP1, and since $\kappa\le n$, it is also
$\Omega(1/\log\log(n/\delta))$-PROP1. This completes the proof.
\end{proof}

\subsection{Proof of Theorem~\ref{thm:value-magnitudes}}
\label[appendix]{app:value-magnitudes}

\thmvaluemagnitudes*

\begin{proof}
Let $m=Ln$. Independently for each agent $i$ and each arrival $t$, set $v_i(g_t)=1$ with
probability $1/2$ and $v_i(g_t)=1/m$ otherwise. Let $T_i$ be the set of goods worth one to
$i$. Since $N_t=[n]$ for every good, the allocation $A$ of an algorithm satisfying the
theorem's assumption is independent of all the sets $T_i$, even when its choices depend
on its earlier allocations.

Conditional on $A$, write $c_i=|A_i|$. Since $\sum_i c_i=Ln$, more than $n/2$ agents have
$c_i\le2L$. For each of these agents, the probability that $A_i\cap T_i=\varnothing$ is
$2^{-c_i}\ge4^{-L}\ge4\delta/n$. These events are independent across agents conditional
on $A$. Hence, the probability that some such agent receives no good in $T_i$ is at least
$1-(1-4\delta/n)^{n/2}\ge1-e^{-2\delta}$. Hoeffding's inequality gives $\Prb[|T_i|<m/4]\le e^{-m/8}\le e^{-n/8}$, and the assumption
on $\delta$ implies $ne^{-n/8}\le\delta/2$. By a union bound, $|T_i|\ge m/4$ for every agent
$i$ except with probability at most $\delta/2$. Therefore, with probability at least
$1-e^{-2\delta}-\delta/2>2\delta-2\delta^2-\delta/2\ge\delta$, some agent $i$ has
$A_i\cap T_i=\varnothing$ and $|T_i|\ge m/4$. Here the strict inequality uses
$e^{-x}<1-x+x^2/2$ for $x>0$, and the last inequality uses $\delta\le1/4$. Such an agent has $v_i(A_i)=c_i/m\le1$, $b_i\le1$, and $V_i/n\ge L/4$,
so $\rho_{\text{PROP1}}(A)\le8/L$. Averaging over the finitely many choices of the sets
$T_i$ gives a fixed input with the stated probability. Finally,
$L\ge(\log_2(n/\delta)-4)/2$, so the ratio is $O(1/\log(n/\delta))$ in the stated range.
\end{proof}

\subsection{Guarantees for All but a Small Fraction of Agents}
\label[appendix]{app:all-but}

The guarantee in \cref{thm:main} must hold for every agent in the same allocation. Allowing
a small fraction of agents to be excluded gives a factor independent of $n$.

\begin{proposition}
\label[proposition]{prop:all-but}
For every input and every $\eps\in[1/n,1)$, with probability at least
$1-2^{-\lceil\eps n\rceil}$, all but at most $\lceil\eps n\rceil-1$ agents satisfy
$\mu_i\le12v_i(A_i)+4B_\eps\nu_i$ in the allocation returned by $\DA$, where
$B_\eps=\log_{4/3}(1+2\ln(4e/\eps))$. Each such agent also satisfies
$v_i(A_i)+b_i\ge\min\{1,n/(\kappa(12+4B_\eps))\}\cdot V_i/n$.
\end{proposition}

\begin{proof}
Recall $H$ from \cref{sec:main-result} and $\phi$ from the proof of \cref{thm:main}.
For every finite set $S$ of encountered agent--scale pairs, \cref{lem:joint-deficit} gives
\[
 \E\left[\prod_{(i,d)\in S}\phi(D_{i,d})\right]\le2^{|S|}.
\]
Indeed, write $\phi(D)=1+\int_0^\infty\phi'(y)\one\{D>y\}\,dy$, expand the product,
and apply the joint bound to each term. The integrals then factor, with each factor
bounded by $1+\int_0^\infty\phi'(y)e^{-H(y)}\,dy=2$.

For $\nu_i>0$, set $Z_i=\sum_d2^dD_{i,d}/(2\nu_i)$; set $Z_i=0$ when $\nu_i=0$.
Convexity expresses $\phi(Z_i)$ as at most a weighted sum of the $\phi(D_{i,d})$ and
$\phi(0)=1$, with nonnegative weights summing to one. For any set $J$ of agents, multiplying
these bounds over $i\in J$ and expanding writes $\prod_{i\in J}\phi(Z_i)$ as at most a
weighted sum, with weights summing to one, of products of $\phi(D_{i,d})$ over sets of
agent--scale pairs with distinct agents. By the product bound, each such product has
expectation at most $2^{|J|}$. Hence, $\E[\prod_{i\in J}\phi(Z_i)]\le2^{|J|}$, and Markov's
inequality implies
\[
 \Prb[Z_i>B_\eps\text{ for all }i\in J]\le\bigl(2e^{-H(B_\eps)/2}\bigr)^{|J|}.
\]
Let $k=\lceil\eps n\rceil$. A union bound over the $k$-element sets of agents bounds
the probability that at least $k$ agents have $Z_i>B_\eps$ by
\[
 \binom nk(2e^{-H(B_\eps)/2})^k
 \le\left(\frac{en}{k}\frac{\eps}{2e}\right)^k\le2^{-k}.
\]
Every other agent satisfies the claimed utility bound by \eqref{eq:main-accounting}.
Using $\nu_i\le v_i(A_i)+b_i$ and $\mu_i\ge V_i/\kappa$ gives the PROP1 inequality.
\end{proof}

For each fixed $\eps>0$, this gives a positive approximation factor independent of $n$,
with success probability tending to one exponentially fast in $n$. The algorithm remains
$\DA$ and does not require $\eps$ as input.

\section{Omitted Proofs in Section~\ref{sec:extensions}}
\label[appendix]{app:sec-extensions}

\subsection{Refined Algorithm and Proof of Theorem~\ref{thm:near-proportional}}
\label[appendix]{app:near-proportional}

For $\eps\in(0,1)$, the refined algorithm $\calA_\eps$ uses scales
$[(1+\eps/2)^d,(1+\eps/2)^{d+1})$ and maintains $L=\lceil8/\eps\rceil$ deficits for each
encountered agent--scale pair, denoted $D_{i,d,1},\dots,D_{i,d,L}$ and initialized to zero.
We refer to these as the $L$ copies of the pair. The algorithm uses the weight function
$w_\eps(y)=\exp((\eps/8)((1+\eps^2/512)^y-1))$. Each copy receives an increase of $x_{it}/L\le\eps/8$, and only the selected copy has its
deficit decreased. This permits the decrease to be close to one while preserving the doubly
exponential bound. The whole good goes to the selected agent; goods themselves are never
divided. As in $\DA$, the selection probabilities use deficits before the current update.
The algorithm depends on $\eps$, but not on $\delta$ or future values.

\thmnearproportional*

\begin{proof}
We first verify that the same joint deficit analysis applies when each decrease is only
$1+\eps/2$. We then sum the resulting bounds over scales and copies.

\begin{algorithm}[t]
\caption{Refined Deficit Allocation ($\calA_\eps$)}
\label{alg:refined-da}
\begin{algorithmic}[1]
\Require Number of agents $n\ge2$ and $\eps\in(0,1)$.
\State Set $L=\lceil8/\eps\rceil$.
\State Initialize $D_{i,d,\ell}=0$ for every $\ell\in[L]$ when $(i,d)$ first occurs.
\For{each arriving good $g_t$}
    \If{$N_t=\varnothing$}
        \State Allocate $g_t$ arbitrarily and continue.
    \EndIf
    \For{each $i\in N_t$}
        \State Set $d_i=\lfloor\log_{1+\eps/2}v_i(g_t)\rfloor$.
    \EndFor
    \State Choose $(j,\ell^*)\in N_t\times[L]$ with probability
    \[
        \frac{(x_{jt}/L)w_\eps(D_{j,d_j,\ell^*})}
             {\sum_{i\in N_t}\sum_{\ell\in[L]}(x_{it}/L)w_\eps(D_{i,d_i,\ell})}.
    \]
    \State Allocate $g_t$ to $j$.
    \For{each $(i,\ell)\in N_t\times[L]$}
        \State $D_{i,d_i,\ell}\gets D_{i,d_i,\ell}+x_{it}/L$.
    \EndFor
    \State $D_{j,d_j,\ell^*}\gets\max\{0,D_{j,d_j,\ell^*}-(1+\eps/2)\}$.
\EndFor
\end{algorithmic}
\end{algorithm}

\paragraph{The deficit bound.}
Recall that $L=\lceil8/\eps\rceil$. Apply \cref{lem:joint-deficit} with $\theta=\eps^2/512$,
$\beta=1+\theta$, $\gamma=64/\eps$, $c=1+\eps/2$, and $\xi=1/L$. Write $H_\eps(y)=\gamma(\beta^y-1)$. Then, $e^{\theta H_\eps(y)}=w_\eps(y)$, exactly the
weight in the algorithm. Each copy receives at most $\xi\le\eps/8$ in a round, and
all fractional amounts across copies sum to one.

The first condition in \eqref{eq:deficit-conditions} holds with equality. To check the
other two, write $s=c-\xi$, so $1+3\eps/8\le s\le1+\eps/2\le3/2$. In particular,
$s\theta\le\eps/8$. Bernoulli's inequality gives
$(1-\theta)\beta^s\ge(1-\theta)(1+s\theta)=1+\theta(s-1-s\theta)\ge1$. The same inequalities
imply $(1-\theta)\gamma(\beta^s-1)\ge\gamma\theta(1-\theta)s\ge(\eps/8)(1+\eps/4)$.
Let $z=(\eps/8)(1+\eps/4)$; then $\tau\le e^{-z}$ and $z\le3\eps/16$. Using
$C\le1+2\theta$, $\ln C\le2\theta$, and $e^{-z}\le1-z+z^2/2$, we obtain
\begin{align*}
 1-\ln C-C\tau
 \ge1-e^{-z}-4\theta
 & \ge z-z^2/2-4\theta\\
 &\ge\frac\eps8+\frac{\eps^2}{32}
            -\frac{9\eps^2}{512}-\frac{\eps^2}{128}\\
 &=\frac\eps8+\frac{3\eps^2}{512}\ge\gamma(\beta-1).
\end{align*}
All conditions of \cref{lem:joint-deficit} hold. In particular, every copy has the marginal
bound $\Prb[D_{i,d,\ell}\ge y]\le e^{-H_\eps(y)}$.

\paragraph{The utility bound.}
For each copy, let $Y_{i,d,\ell}$ be its total fractional amount and $N_{i,d,\ell}$ its number
of selections. The updates give $Y_{i,d,\ell}\le(1+\eps/2)N_{i,d,\ell}+D_{i,d,\ell}$. The
fractional amounts sum over copies to those of the original benchmark. Moreover, every
selection corresponds to a distinct good given to the relevant agent, with value at least
$(1+\eps/2)^d$. Therefore
\begin{equation}
 \begin{aligned}
 \mu_i
 &\le(1+\eps/2)\sum_{d,\ell}(1+\eps/2)^dY_{i,d,\ell}\\
 &\le(1+\eps/2)^2v_i(A_i)+(1+\eps/2)\sum_{d,\ell}(1+\eps/2)^dD_{i,d,\ell}.
 \end{aligned}
 \label{eq:refined-accounting}
\end{equation}
For $\nu_i>0$, the sum of the coefficients over all encountered scales and copies is at most
\[
 L\sum_{d:\,(1+\eps/2)^d\le\nu_i}(1+\eps/2)^d
 \le\frac{2L(1+\eps/2)}{\eps}\nu_i.
\]
Define
\[
 B^{(\eps)}_{n,\delta}
 =\frac{\ln\left(1+\frac\eps{32}\ln(2n/\delta)\right)}
        {\ln(1+\eps^2/512)},
\]
so that $H_\eps(B^{(\eps)}_{n,\delta})=2\ln(2n/\delta)$. The argument in
\eqref{eq:deficit-moment}--\eqref{eq:weighted-deficit-tail} uses only the marginal deficit
bound and convexity of $e^{H(y)/2}$. Applying it with $H_\eps$ and the coefficient sum above
gives, except with probability $\delta/n$,
\[
 \sum_{d,\ell}(1+\eps/2)^dD_{i,d,\ell}
 \le\frac{2L(1+\eps/2)}{\eps}B^{(\eps)}_{n,\delta}\nu_i.
\]
An agent with $\nu_i=0$ satisfies the desired inequality deterministically. Taking a union
bound over agents and substituting into \eqref{eq:refined-accounting}, we obtain, with
probability at least $1-\delta$, for all $i\in[n]$,
\begin{equation}
 v_i(A_i)\ge\frac{\mu_i}{(1+\eps/2)^2}
       -\frac{2L}{\eps}B^{(\eps)}_{n,\delta}\nu_i.
 \label{eq:refined-explicit}
\end{equation}
Finally, $(1+\eps/2)^{-2}\ge1-\eps$ and
$B^{(\eps)}_{n,\delta}\le B_{n,\delta}\ln(4/3)/\ln(1+\eps^2/512)$, where $B_{n,\delta}$ is as
in \cref{thm:main}. Therefore, with
\[
C_\eps=\frac{2\lceil8/\eps\rceil\ln(4/3)}{\eps\ln(1+\eps^2/512)},
\]
\eqref{eq:refined-explicit}
implies that, with probability at least $1-\delta$, every agent $i$ satisfies
\begin{equation}
 v_i(A_i)\ge(1-\eps)\mu_i-C_\eps B_{n,\delta}\nu_i.
 \label{eq:near-prop-explicit}
\end{equation}
Since $C_\eps$ depends only on $\eps$ and $B_{n,\delta}\le11\log\log(n/\delta)$ by the proof of
\cref{thm:main}, this proves \cref{thm:near-proportional}.
\end{proof}

\subsection{Proof of Corollary~\ref{cor:small-goods}}
\label[appendix]{app:small-goods}

\corsmallgoods*

\begin{proof}
By \eqref{eq:near-prop-explicit} in \cref{app:near-proportional}, applied with accuracy
$\eps/2$, with probability at least $1-\delta$, every agent $i$ satisfies
$v_i(A_i)\ge(1-\eps/2)\mu_i-C_{\eps/2}B_{n,\delta}\nu_i$, where $C_{\eps/2}$ depends only on $\eps$
and $B_{n,\delta}$ is as in \cref{thm:main}. Let $K_\eps=22C_{\eps/2}/\eps$. Since
$B_{n,\delta}\le11\log\log(n/\delta)$ by the proof of \cref{thm:main}, the condition
$V_i/n\ge K_\eps\nu_i\log\log(n/\delta)$ gives $C_{\eps/2}B_{n,\delta}\nu_i\le\eps V_i/(2n)$.
Since $\mu_i\ge V_i/n$, it follows that
$v_i(A_i)\ge(1-\eps/2)\mu_i-\eps V_i/(2n)\ge(1-\eps)V_i/n$ for all agents on the same event of
probability at least $1-\delta$.
\end{proof}

\subsection{Proof of Theorem~\ref{thm:propk}}
\label[appendix]{app:propk-proof}

\thmpropk*

\begin{proof}
Let $B_{n,\delta}$ be as in \cref{thm:main}, and fix an agent $i$. If she values
fewer than $k$ goods positively, then $v_i(A_i)+b_i^{(k)}=V_i\ge\mu_i$, so both conclusions hold
deterministically for her. Otherwise, let $a_i>0$ be her $k$th largest value among all goods,
and let $G_i$ be the set of goods that she values at a scale $d$ with $2^d>a_i$. This set
depends only on the fixed input. Every good in $G_i$ is worth strictly more than $a_i$ to her,
so $|G_i|\le k-1$.

Applying the accounting in \eqref{eq:deficit-accounting} only to the scales $d$ with
$2^d\le a_i$, which contain all goods outside $G_i$ that agent $i$ values positively, gives
\begin{equation}
 \sum_{t:g_t\notin G_i}x_{it}v_i(g_t)
 \le12v_i(A_i\setminus G_i)+2\sum_{d:\,2^d\le a_i}2^dD_{i,d}.
 \label{eq:propk-lower-scales}
\end{equation}
The coefficients $2^d$ in the last sum add up to at most $2a_i$. Consequently,
\eqref{eq:deficit-moment} and convexity, with weights $2^d/(2a_i)$ and any remaining
weight on zero, give $\E[\phi(\sum_{d:\,2^d\le a_i}2^dD_{i,d}/(2a_i))]\le2$.
As in \eqref{eq:weighted-deficit-tail}, Markov's inequality implies
$\sum_{d:\,2^d\le a_i}2^dD_{i,d}\le2B_{n,\delta}a_i$ with probability at least $1-\delta/n$.
All deficit sums here range over encountered scales only.

Since $x_{it}\le1$ and $|G_i|<k$, the contribution of the larger goods satisfies
\begin{equation}
 \sum_{t:g_t\in G_i}x_{it}v_i(g_t)
 \le v_i(A_i\cap G_i)+b_i^{(k)}.
 \label{eq:propk-large-goods}
\end{equation}
Also, the $k$ most valuable goods have total value at least $ka_i$. Those assigned to $i$
contribute at most $v_i(A_i)$, and the others form a set of at most $k$ outside goods. Thus
$ka_i\le v_i(A_i)+b_i^{(k)}$. Combining these observations with \eqref{eq:propk-lower-scales}, we obtain, with probability
at least $1-\delta/n$,
\[
 \mu_i\le12v_i(A_i)+b_i^{(k)}+4B_{n,\delta}a_i
 \le\left(12+\frac{4B_{n,\delta}}k\right)(v_i(A_i)+b_i^{(k)}).
\]
By a union bound over agents and $\mu_i\ge V_i/\kappa$, the allocation is
$\min\{1,n/(\kappa(12+4B_{n,\delta}/k))\}$-PROP$k$ with probability at least $1-\delta$. Since
$12+4B_{n,\delta}/k\le16\max\{1,B_{n,\delta}/k\}$ and $n/\kappa\ge1$, this factor is at least
$\frac1{16}\min\{1,nk/(\kappa B_{n,\delta})\}$. As $B_{n,\delta}\le11\log\log(n/\delta)$, this
proves the first statement.

For the refined algorithm, the same argument uses its narrower scales and deficit bound.
The calculation in \cref{app:propk-refined}, applied with accuracy $\eps/2$, shows that, with
probability at least $1-\delta$, every agent satisfies \eqref{eq:refined-propk-comparison} with
$\eps/2$ in place of $\eps$, where $C_{\eps/2}$ is the constant in
\eqref{eq:near-prop-explicit}. Since $K_\eps=22C_{\eps/2}/\eps$ and
$B_{n,\delta}\le11\log\log(n/\delta)$, the condition $k\ge K_\eps\log\log(n/\delta)$ gives
$C_{\eps/2}B_{n,\delta}/k\le\eps/2$, and therefore
$v_i(A_i)+b_i^{(k)}\ge((1-\eps/2)/(1+\eps/2))\mu_i\ge(1-\eps)\mu_i$. This proves the second
statement.
\end{proof}

\subsection{Refined Comparison in the Proof of Theorem~\ref{thm:propk}}
\label[appendix]{app:propk-refined}

We verify the comparison used for $\calA_{\eps/2}$ in the proof of \cref{thm:propk}.
First consider $\calA_\eps$ for any $\eps\in(0,1)$, retaining $L$, $H_\eps$, and
$B^{(\eps)}_{n,\delta}$ from \cref{app:near-proportional}. Fix an agent $i$ who values at least
$k$ goods positively. Let $a_i$ be her $k$th largest value, and let $G_i$ consist of the goods that she
values at a scale $d$ with $(1+\eps/2)^d>a_i$. As in the proof above, $|G_i|<k$,
\eqref{eq:propk-large-goods} applies, and $ka_i\le v_i(A_i)+b_i^{(k)}$.

Restrict \eqref{eq:refined-accounting} to the scales $d$ with $(1+\eps/2)^d\le a_i$. Their
total coefficient, including the $L$ copies, is at most
\[
 \sum_{d:\,(1+\eps/2)^d\le a_i}\ \sum_{\ell\in[L]}(1+\eps/2)^d
 \le\frac{2L(1+\eps/2)}{\eps}a_i.
\]
The convexity and Markov argument preceding \eqref{eq:refined-explicit} therefore gives,
with probability at least $1-\delta/n$,
\[
 \sum_{d:\,(1+\eps/2)^d\le a_i}\ \sum_{\ell\in[L]}(1+\eps/2)^dD_{i,d,\ell}
 \le\frac{2L(1+\eps/2)}{\eps}B^{(\eps)}_{n,\delta}a_i.
\]
Substituting into the restricted accounting inequality and dividing by $(1+\eps/2)^2$ gives us
\begin{align*}
 v_i(A_i\setminus G_i)
 &\ge\frac{1}{(1+\eps/2)^2}
       \sum_{t:g_t\notin G_i}x_{it}v_i(g_t)
       -\frac{2L}{\eps}B^{(\eps)}_{n,\delta}a_i\\
 &\ge(1-\eps)\sum_{t:g_t\notin G_i}x_{it}v_i(g_t)
       -C_\eps B_{n,\delta}a_i.
\end{align*}
Here we used the same inequalities and the same choice of $C_\eps$ as in
\eqref{eq:near-prop-explicit}. Adding \eqref{eq:propk-large-goods} in the form
$v_i(A_i\cap G_i)+b_i^{(k)}\ge\sum_{t:g_t\in G_i}x_{it}v_i(g_t)$ gives
$v_i(A_i)+b_i^{(k)}\ge(1-\eps)\mu_i-C_\eps B_{n,\delta}a_i$. Since
$ka_i\le v_i(A_i)+b_i^{(k)}$, this implies
\begin{equation}
 \left(1+\frac{C_\eps B_{n,\delta}}k\right)(v_i(A_i)+b_i^{(k)})
 \ge(1-\eps)\mu_i.
 \label{eq:refined-propk-comparison}
\end{equation}
An agent with fewer than $k$ positively valued goods satisfies this inequality
deterministically. A union bound proves it for every agent with probability at least
$1-\delta$. Applied with accuracy $\eps/2$, this is the comparison used in the proof of
\cref{thm:propk}.

\subsection{Proof of Theorem~\ref{thm:fractional-lower}}
\label[appendix]{app:fractional-lower}

We adapt the construction of \citet[Theorem~4]{gaoEtAl2021improved}. Its key feature is
that the algorithm must allocate a good before learning which of its two interested
agents will value later goods. We give the construction and its probability bound within
the proof, and then derive the fairness consequences.

\thmfractionallower*

\begin{proof}
For each integer $j\ge0$, construct a random sequence $T_j$ on $2^j$ agents, together with
one distinguished agent. The sequence $T_0$ has no goods, and its only agent is distinguished.
For $j\ge1$, divide the agents into two equal sets. Present an independent copy of
$T_{j-1}$ on the first set and then one on the second set. If their distinguished agents
are $a$ and $b$, present one more good worth one to $a$ and $b$ and zero to everyone else.
Choose one of $a$ and $b$ uniformly to be the distinguished agent of $T_j$.
All choices are sampled independently before the algorithm runs, so the input is oblivious.
Although sampled in advance, a choice need not be revealed to the algorithm until it determines
the value of a later good.
The sequence has $2^j-1$ goods; every good is valued by exactly two agents, and the
distinguished agent values exactly $j$ goods, each at one.

For every deterministic online algorithm and every history preceding this sequence, we
claim that its distinguished agent receives none of her valued goods with conditional
probability at least $q_j=2^{-(2^j-1)}$. We induct on $j$. The claim holds for $j=0$. For the induction step, the corresponding
probability in the first copy is at least $q_{j-1}$. After fixing all the random choices
and allocations of that copy, including its distinguished agent, the second copy still
has probability at least $q_{j-1}$ by induction. Thus both distinguished agents receive
none of their valued goods within their copies with probability at least $q_{j-1}^2$.
The final choice between $a$ and $b$ is independent of the recipient of their common
good. At least one of them does not receive that good, even if the algorithm gives it to
an agent who values it at zero. The required probability is therefore at least
$q_{j-1}^2/2=q_j$, proving the claim.

For positive integers $G$ and $\ell$ with $G2^\ell\le n$, present $G$ copies of
$T_\ell$ on disjoint sets of agents, one after another. The other agents value every good at zero. Since the claim holds after
every history, the probability that none of the $G$ distinguished agents receives zero
utility is at most $(1-q_\ell)^G$. Consequently, with probability at least
$1-(1-q_\ell)^G$, the returned allocation has an agent $i$ satisfying
\begin{equation}
 \begin{gathered}
 v_i(A_i)=0,\quad \mu_i=\ell/2,\quad \nu_i=1,\\
 b_i^{(k)}=\min\{k,\ell\}\quad\text{for every positive integer }k.
 \end{gathered}
 \label{eq:fractional-lower-agent}
\end{equation}
Receiving goods from other copies cannot change this agent's utility, since she values
them at zero. The existence of an agent satisfying \eqref{eq:fractional-lower-agent}
depends only on the input and allocation, not on which agent is distinguished in the
analysis. Fixing a randomized algorithm's random choices and then averaging therefore
gives the same lower bound for that algorithm on this random input. Some fixed input
attains at least that probability.

Take $\ell=\lfloor\log_2(\tfrac12\log_2(n/\delta))\rfloor$, put $K=2^\ell$, and let
$G=\lfloor n/K\rfloor$.
Then, $\ell\ge1$, $\ell\ge\log_2\log_2(n/\delta)-2$, and $K\le\tfrac12\log_2(n/\delta)\le n$.
In particular, $G\ge n/(2K)$, so $Gq_\ell\ge(n/K)2^{-K}$. Since $\log_2K\le K-1$, we have
$\log_2(n/K)-K\ge\log_2n-2K+1\ge\log_2(2\delta)$, and hence $Gq_\ell\ge2\delta$.
It follows that the probability of \eqref{eq:fractional-lower-agent} on the fixed input
is at least $1-(1-q_\ell)^G\ge1-e^{-Gq_\ell}\ge1-e^{-2\delta}>\delta$, where the last
inequality holds for $0<\delta\le1/2$.
On this event, $v_i(A_i)+b_i=1=2\mu_i/\ell$. No ratio larger than $2/\ell$ relative to
$\mu_i$ can therefore be guaranteed on every input with probability at least $1-\delta$.
Similarly, $v_i(A_i)+b_i^{(k)}\ge(1-\eps)\mu_i$ requires
$k\ge(1-\eps)\ell/2$, and $v_i(A_i)\ge(1-\eps)\mu_i-C\nu_i$ requires
$C\ge(1-\eps)\ell/2$. If either condition is violated, its guarantee is violated whenever
\eqref{eq:fractional-lower-agent} holds, an event of probability greater than $\delta$.
Since $\ell=\Theta(\log\log(n/\delta))$, these give the claimed bounds.
Conversely, \eqref{eq:main-utility} and $\nu_i\le v_i(A_i)+b_i$ show that $\DA$ attains
a ratio of at least $1/(12+4B_{n,\delta})$ for this benchmark with probability at least
$1-\delta$. This gives the matching lower bound on the optimal ratio.
\end{proof}

The input has at most $n-1$ goods, and every agent values at most $\ell\le\log_2 n$ of
them. Each agent with positive total value can reach value at least one after adding one
outside good, while $V_i/n\le\ell/n\le1$. Thus every allocation of this input is PROP1;
the lower bound concerns $\mu_i$, not $V_i/n$.

\paragraph{Expected ratios.}
To state the expected guarantee precisely, let
\[
 \rho_\mu(A)=\min\left(\{1\}\cup
 \left\{\frac{v_i(A_i)+b_i}{\mu_i}:\mu_i>0\right\}\right).
\]
As in \cref{sec:preliminaries}, write $\rho_\mu(\calA,I)$ for this ratio on the random
allocation returned by $\calA$. We show that
\[
 \sup_{\calA}\inf_I\E[\rho_\mu(\calA,I)]
 =\Theta(1/\log\log n).
\]
The lower bound follows from \eqref{eq:main-utility} with $\delta=1/2$ and
$\nu_i\le v_i(A_i)+b_i$.
For the upper bound, set
$\ell=\lfloor\log_2(\tfrac12\log_2 n)\rfloor$ for $n\ge256$,
$K=2^\ell$, and $G=\lfloor n/K\rfloor$.
Then, $K\le\tfrac12\log_2 n$ and $Gq_\ell\ge(n/K)2^{-K}\ge2\sqrt n/\log_2 n$.
Every deterministic algorithm therefore has expected ratio on this input distribution
at most $2/\ell+(1-q_\ell)^G\le2/\ell+\exp(-2\sqrt n/\log_2n)=O(1/\log\log n)$.
Averaging again gives a fixed input with this expected upper bound for every randomized
algorithm. This proves the claimed order.

\subsection{Failure of Ex-Ante Proportionality for Deficit Allocation}
\label[appendix]{app:ex-ante}

\begin{proposition}
\label[proposition]{prop:not-ex-ante}
The algorithm $\DA$ is not ex-ante proportional, even for two agents and two goods.
\end{proposition}

\begin{proof}
Consider two agents and two goods. Both agents value $g_1$ at $1$, and agents $1$ and $2$
value $g_2$ at $2$ and $1$, respectively. Each agent receives the first good with probability $1/2$. Agent $1$ values the second good
at a new scale, so her relevant deficit for that good is zero in either case. If agent $1$
received the first good, then agent $2$'s deficit is $1/2$, and since $w(0)=1$, agent $1$
receives the second with probability $1/(1+w(1/2))$. If agent $2$ received the first good, then
both relevant
deficits are zero, and agent $1$ receives the second with probability $1/2$. Therefore
\begin{align*}
 \E[v_1(A_1)]
 &=\frac12+2\left(\frac{1}{2(1+w(1/2))}+\frac14\right)\\
 &=1+\frac{1}{1+w(1/2)}<\frac32=\frac{V_1}{2},
\end{align*}
where the strict inequality follows from $w(1/2)>1$.
\end{proof}

\subsection{Weighted Entitlements}
\label[appendix]{app:weighted}

Suppose that agent $i$ has entitlement $w_i>0$, where $\sum_iw_i=1$, and that these
entitlements are known to the algorithm. Her \emph{weighted proportional share} is $w_iV_i$.
An allocation satisfies \emph{weighted $\alpha$-PROP$k$} if
$v_i(A_i)+b_i^{(k)}\ge\alpha w_iV_i$ for every agent $i$.
The analysis extends by dividing each good in proportion to the entitlements of the
agents who value it positively.

In this subsection, replace the fractional amounts by $x_{it}=\frac{w_i}{\sum_{j\in N_t}w_j}$ for
$i\in N_t$ and $x_{it}=0$ for $i\notin N_t$, and retain $\mu_i=\sum_t x_{it}v_i(g_t)$. Thus $\mu_i\ge w_iV_i$.
Use these amounts in both the selection probabilities and the deficit increases of
$\DA$ and $\calA_\eps$; all other steps are unchanged.

\begin{corollary}
\label[corollary]{cor:weighted}
For every $\delta\in(0,1)$, every positive integer $k$, and every input, each of the
following guarantees holds with probability at least $1-\delta$.
\begin{enumerate}
\item The weighted version of $\DA$ defined above gives weighted PROP$k$ with factor
$1/(12+4B_{n,\delta}/k)$. For $k=1$, this factor is
$\Omega(1/\log\log(n/\delta))$.
\item For every $\eps\in(0,1)$, the weighted version of $\calA_\eps$ gives
$v_i(A_i)\ge(1-\eps)\mu_i-O_\eps(\nu_i\log\log(n/\delta))$ for every agent $i$.
\item For every $\eps\in(0,1)$, if $k\ge K_\eps\log\log(n/\delta)$, the weighted version
of $\calA_{\eps/2}$ gives weighted $(1-\eps)$-PROP$k$.
\end{enumerate}
The constants do not depend on the entitlements.
\end{corollary}

\begin{proof}
The new fractional amounts are fixed by the input, belong to $[0,1]$, and sum to one
whenever $N_t\ne\varnothing$. Thus \cref{lem:joint-deficit} applies with exactly the
parameters in \cref{app:simple-parameters,app:near-proportional}. Each copy in the
refined algorithm still receives at most $1/L$ in each round.
The accounting and convexity arguments in \cref{app:main-proof,app:near-proportional}
therefore give the same utility bounds for the new $\mu_i$.
The proof in \cref{app:propk-proof} uses only $x_{it}\le1$ and these deficit bounds, so
it again gives $\mu_i\le(12+4B_{n,\delta}/k)(v_i(A_i)+b_i^{(k)})$
for the variant of $\DA$. The refined comparison \eqref{eq:refined-propk-comparison}
is unchanged as well. Finally, $\mu_i\ge w_iV_i$ converts these bounds into the stated
weighted proportionality guarantees.
\end{proof}

\section{Omitted Proofs in Section~\ref{sec:lower-bounds}}
\label[appendix]{app:sec-lower-bounds}

\subsection{Proof of Theorem~\ref{thm:ef1-confidence}}
\label[appendix]{app:ef1-proof}

\thmefoneconfidence*

\begin{proof}
We prove the bound $4e^{-q/128}$ for inputs with $n+q$ goods. Fix a positive multiple $q$ of $16$ with
$n\ge1024q$, fix $\alpha$, and set $r=8q$, $\eta=\alpha/(4r)$, $\zeta=\alpha\eta/(4q)$, and
$K=4/\alpha$. These choices ensure the inequalities $q\zeta<\alpha\eta$, $r\eta+q\zeta<\alpha$,
and $1+r\eta+q\zeta<\alpha K$, which we use to compare utilities in the proof. Choose a uniform $r$-element set $S\subseteq[n]$ of \emph{target agents}. Independently,
assign to each agent $i\notin S$ a uniform group label $\ell_i\in[q]$, and choose independent
fair bits $X_1,\dots,X_q$. All these choices are made before the run. Present the following goods.
\begin{enumerate}
\item $n-r$ \emph{unit goods}, each worth $1$ to every agent.
\item One \emph{targeted good} $g_a$ for each $a\in S$, in increasing order of $a$, worth
$K$ to $a$ and $\eta$ to every other agent.
\item Exactly $q$ final goods $f_1,\dots,f_q$. If $X_j=0$, then $f_j$ is worth $\zeta$ to
every agent. If $X_j=1$, then $f_j$ is worth $K$ to every agent in $S$, $1$ to every agent
$i\notin S$ with $\ell_i=j$, and $\zeta$ to every agent $i\notin S$ with $\ell_i\ne j$.
Call $f_j$ \emph{active} when $X_j=1$.
\end{enumerate}
The first $n$ goods, which we call \emph{initial goods}, reveal neither the labels nor the bits. We first analyze these values and
rescale them to $(0,1]$ at the end of the proof.

Fix a deterministic rule. Let $E$ be the set of agents receiving no unit good, and let
$C=[n]\setminus E$. Since $|E|\ge r$, choose an $r$-element subset $E_0\subseteq E$ determined
only by the allocation of the unit goods, and write $k=|S\cap E_0|$. Since the unit goods arrive
first and are the same on every input, the set $E_0$ is independent of $S$. The estimate in
\cref{app:ef1-estimates} gives $\Prb[k>q/8]\le e^{-q/128}$. We henceforth assume that $k\le q/8$.

\paragraph{The unit goods must have distinct recipients.}
Suppose otherwise, so that some bundle contains two unit goods. Every agent $i\in E\setminus S$
must then receive $f_{\ell_i}$, and this good must be active: otherwise her utility is at most
$r\eta+q\zeta<\alpha$, while removing any one good from the bundle with two unit goods leaves
value at least $1$ to her. But $|E\setminus S|\ge|E_0\setminus S|=r-k>q$, and there are at most
$q$ active goods. Thus the final allocation cannot be $\alpha$-EF1.
We may therefore assume that $|E|=r$ and each agent in $C$ has exactly one unit good.
Then, $E_0=E$, so $k=|S\cap E|$.

\paragraph{Many agents without initial goods.}
After the targeted goods, let $Z\subseteq E$ be the agents who received no targeted good; these
agents have no initial good. Write $z=|Z|$ and $u=|Z\setminus S|$. Then, $z\le u+k$.
For each fixed $S$, these sets and counts are determined before any label or bit is revealed.
Every agent $i\in Z\setminus S$ must receive the active good $f_{\ell_i}$ if the final
allocation is $\alpha$-EF1. Otherwise $v_i(A_i)\le q\zeta<\alpha\eta$. Since the first $n$
goods went to the other $n-1$ agents, some other bundle contains two of them, and removing any
one good from that bundle leaves value at least $\eta$ to $i$, contradicting $\alpha$-EF1.

Consequently, the $u$ agents in $Z\setminus S$ must have distinct labels, and all corresponding
bits must equal one. Given $S$ and the labels, the bits are still independent and fair, so this
event has probability either zero or $2^{-u}$. In particular,
$\Prb[\text{success and }u\ge q/16\mid S]\le2^{-q/16}$.

\paragraph{Few agents without initial goods.}
It remains to consider $u<q/16$. Let $b$ be the number of groups containing no agent in
$E\setminus S$, and let $s=\sum_{j=1}^qX_j$ be the number of active goods. We show that
$s\le u+2k+2b$ whenever the final allocation is $\alpha$-EF1.

We first prove that active goods have distinct recipients. Each of the $r-z$ agents in
$E\setminus Z$ received a targeted good, so at most $z$ targeted goods went to $C$, and hence
at most $z$ agents in $S\cap C$ received their own targeted good. Moreover, at most $q$ targets
receive an active good. Since $|S\cap C|-z-q=r-k-z-q\ge r-2k-u-q>0$, some target
$a\in S\cap C$ receives neither $g_a$ nor any active good. Her utility is at most
$1+r\eta+q\zeta<\alpha K$. Every active good is worth $K$ to her, so a bundle containing two
active goods would contradict $\alpha$-EF1.

Next, at most $b$ active goods can go to $C$. Suppose that active $f_j$ goes to an agent in
$C$ and group $j$ contains an agent $i\in E\setminus S$. Agent $i$ does not receive
$f_j$, so her utility is at most $r\eta+q\zeta<\alpha$: all other final goods are worth
$\zeta$ to her. The recipient of $f_j$ also has a unit good, and $i$ values both goods at
$1$. This contradicts $\alpha$-EF1. Hence, an active good can go to $C$ only if its group is
one of the $b$ groups containing no agent in $E\setminus S$.

Finally, at most $z+k+b$ active goods can go to $E$. Consider a recipient in $E$ who also
holds $g_a$ for some $a\in S\cap C$. The target $a$ must receive an active good as well:
otherwise her utility is at most $1+r\eta+q\zeta<\alpha K$, while the recipient's bundle
contains $g_a$ and an active good, both worth $K$ to $a$. Distinct such recipients correspond
to distinct targets in $C$, since each target has only one targeted good. Each of these
targets receives an active good, and we showed above that at most $b$ active goods go to $C$.
There are therefore at most $b$ such recipients. Every other recipient in $E$ either belongs to $Z$, accounting
for at most $z$ agents, or received a nonempty set of targeted goods, all with targets in $E$.
There are only $k$ such goods, so at most $k$ agents of the latter kind. Since active
goods have distinct recipients, at most $z+k+b$ go to $E$.

Adding the bounds for $C$ and $E$, and using $z\le u+k$, gives $s\le u+2k+2b$. When also
$b<q/16$, this implies $s<q/16+q/4+q/8=7q/16$, which is substantially below the expected
number $q/2$ of active goods.

\paragraph{Combining the probabilities.}
For each fixed $S$ with $k\le q/8$, the $r-k\ge63q/8$ agents in $E\setminus S$ have
independent uniform group labels. The estimates in \cref{app:ef1-estimates} give
$\Prb[b\ge q/16\mid S]\le e^{-q/8}$ and $\Prb[s\le7q/16]\le e^{-q/128}$.
If the final allocation is $\alpha$-EF1 and $k\le q/8$, then the unit goods have distinct
recipients, and the preceding paragraphs show that $u\ge q/16$, $b\ge q/16$, or $s<7q/16$.
Together with the bounds on $\Prb[k>q/8]$ and on $\Prb[\text{success and }u\ge q/16\mid S]$
derived above, these bounds imply
\[
 \Prb[\text{success}]
 \le e^{-q/128}+2^{-q/16}+e^{-q/8}+e^{-q/128}
 \le4e^{-q/128}.
\]
This proves the bound for every deterministic rule. Averaging this bound over the internal randomness of $\calA$ gives the same bound for
$\calA$ on the random input. Averaging over the input distribution then yields a fixed input
on which $\calA$ satisfies the same bound.

To put all values in $(0,1]$, divide every value in the family by $K$. The four levels become
$\alpha^3/(64rq)$, $\alpha^2/(16r)$, $\alpha/4$, and $1$. Since $\alpha$-EF1 is unaffected by
multiplying all values by the same positive number, the argument above applies to every
deterministic rule on the scaled family. Moreover, every input in the family has exactly
$n+q$ goods, so knowing the number of goods does not help. This completes the proof.
\end{proof}

\subsection{Probability Estimates in the Proof of Theorem~\ref{thm:ef1-confidence}}
\label[appendix]{app:ef1-estimates}

We use the notation of the proof of \cref{thm:ef1-confidence}. The allocation of the unit
goods is fixed by a deterministic rule, $r=8q$, and $n\ge1024q$.

\paragraph{The number of targets without unit goods.}
The set $E_0$ has size $r$ and does not depend on the uniform target set $S$. Expanding
$2^k=2^{|S\cap E_0|}$ over subsets of $E_0$ gives
\begin{align*}
 \E[2^{k}]
 &=\sum_{T\subseteq E_0}\Prb[T\subseteq S]
 \le\sum_{T\subseteq E_0}(r/n)^{|T|}\\
 &=(1+r/n)^r\le e^{r^2/n}\le e^{q/16}.
\end{align*}
The first inequality holds because
$\Prb[T\subseteq S]=\binom{n-|T|}{r-|T|}/\binom nr\le(r/n)^{|T|}$. Markov's inequality applied
to $2^{k}$ now yields
\[
 \Prb[k>q/8]
 \le\exp\left(-\frac{2\ln2-1}{16}q\right)
 \le e^{-q/128},
\]
which is the bound used in \cref{app:ef1-proof}.

\paragraph{Groups containing no agent in $E\setminus S$.}
Suppose the unit goods have distinct recipients, and fix $S$ with $k=|S\cap E|\le q/8$.
The $r-k\ge63q/8$ agents in $E\setminus S$ have independent uniform labels in $[q]$.
If $b\ge q/16$, there is a set of $q/16$ groups containing none of these agents. For any
fixed set of that size, the probability is $(15/16)^{r-k}$. A union bound therefore gives
\begin{align*}
 \Prb[b\ge q/16\mid S]
 &\le\binom q{q/16}(15/16)^{63q/8}\\
 &\le\exp\left(q\left[
       \frac{\ln(16e)}{16}+\frac{63}{8}\ln\frac{15}{16}
       \right]\right)
 \le e^{-q/8}.
\end{align*}
For the last inequality, $\ln2\le3/4$ and $\ln(1-1/16)\le-1/16$ make the coefficient
in brackets at most $-31/128<-1/8$. The estimate is uniform over all the fixed sets $S$
under consideration.

\paragraph{The number of active goods.}
Since $s=\sum_{j=1}^qX_j$ and the bits are independent and fair, Hoeffding's inequality gives
$\Prb[s\le7q/16]=\Prb[s-q/2\le-q/16]\le e^{-q/128}$. For a direct derivation, apply Markov's
inequality to $\exp(\lambda(q/2-s))$ and use
$\E[\exp(\lambda(q/2-s))]=\cosh(\lambda/2)^q\le e^{q\lambda^2/8}$. Taking $\lambda=1/4$ gives the stated bound.

\subsection{Proof of Corollary~\ref{cor:ef1-expected}}
\label[appendix]{app:ef1-expected}

\corefoneexpected*

\begin{proof}
For $n\ge16384$, take $q=16\lfloor n/16384\rfloor$. For each $\alpha\in(0,1)$, the proof of
\cref{thm:ef1-confidence} gives an input $I_\alpha$ on which
\begin{align*}
 \E[\rho_{\text{EF1}}(\calA,I_\alpha)]
 &\le\alpha+(1-\alpha)\Prb[\rho_{\text{EF1}}(\calA,I_\alpha)\ge\alpha]\\
 &\le\alpha+4(1-\alpha)e^{-q/128}.
\end{align*}
Taking the infimum over inputs and then letting $\alpha\downarrow0$ gives
$R^{\calA}_{\text{EF1}}(n)\le4e^{-q/128}=e^{-\Omega(n)}$, proving the corollary.
\end{proof}

\subsection{Proof of Theorem~\ref{thm:mms-confidence}}
\label[appendix]{app:mms-proof}

\thmmmsconfidence*

\begin{proof}
Put $r=\lfloor n/2\rfloor$ and $K=2/\alpha$.
Before the run, choose a \emph{target} agent $a\in[n]$ uniformly at random and, independently, choose
each of the two endings with probability $1/2$. Every input starts with $r$ unit goods
worth $1$ to every agent, followed by $n-r$ targeted goods worth $K$ to $a$ and $\alpha/2$ to
every other agent. The final $n-r$ goods are either all worth zero, or each worth $K$ to
$a$ and $1$ to every other agent. All $2n$ inputs have the same number of goods, and all choices
defining the input are made before any allocation is observed.

Fix a deterministic rule. Its allocation of the unit goods is independent of the target and the
ending. For each target, its allocation of the targeted goods is the same on both inputs.

The input with the zero-valued ending has exactly $n$ goods that every agent values
positively. Placing one such good in each of $n$ bundles shows that every agent has a strictly
positive maximin share.  Thus, on the input with the zero-valued ending, every agent must
receive one of these $n$ goods in any successful allocation. If some agent receives two unit
goods, then no input with the zero-valued ending can succeed; only the $n$ inputs with the
positive-valued ending can succeed, giving success probability at most $1/2\le1-r/(2n)$.

Otherwise, let $C$ be the set of the $r$ recipients of the unit goods, and let $E=[n]\setminus C$.
Fix a target $a\in C$, and suppose that the input with the zero-valued ending succeeds.
Then, the $n-r$ targeted goods must be allocated one each to the agents in $E$.  In particular, the target $a$
receives no targeted good.

Now consider the positive-valued ending with the same target. For every agent $i\in E$,
the $r$ unit goods and the $n-r$ final goods are all worth $1$, so $\mathsf{MMS}_i\ge1$.  Her utility before the final
goods arrive is only $\alpha/2<\alpha$.  Therefore she must receive a final good in any
$\alpha$-MMS allocation.  There are exactly $|E|=n-r$ final goods, so they must all go
to $E$, leaving the target $a$ with utility $1$.

On the other hand, agent $a$ values all $2(n-r)\ge n$ targeted and final goods at $K$.
A partition into $n$ bundles, each containing at least one of these goods, shows that
$\mathsf{MMS}_a\ge K$.  Hence, $v_a(A_a)=1<2=\alpha K\le\alpha\cdot\mathsf{MMS}_a$,
so the allocation on this input is not $\alpha$-MMS.  For every target in $C$, at most one of its two inputs can
therefore succeed.  Allowing both inputs to succeed for each target outside $C$, at most $r+2(n-r)=2n-r$ of the $2n$ inputs can succeed.  This gives the bound $1-r/(2n)$ for every deterministic
rule.  Averaging this bound over the internal randomness of $\calA$ gives the same bound for
$\calA$ on the random input. Averaging over the input distribution then yields a fixed input
on which $\calA$ satisfies the same bound.  Finally, a common scaling of the whole family by $1/K$ puts
all values in $[0,1]$ without changing the argument.
\end{proof}

The bound in \cref{thm:mms-confidence} does not deteriorate as $\alpha$ tends to zero.
This differs from the classical upper bound on the expected realized guarantee
$R^{\calA}_{\text{MMS}}(n)=O(n^{-1/2})$ that follows, for identical valuations, from the randomized
machine-covering lower bound of \citet{azarEpstein1998covering}.  The latter only gives a
success-probability bound of $O(1/(\alpha\sqrt n))$, which becomes vacuous for sufficiently
small $\alpha$.  Thus \cref{thm:mms-confidence} bounds the success probability uniformly over all positive
factors, including the range in which the bound on the expected realized guarantee is vacuous.

\subsection{Proof of Theorem~\ref{thm:efx-confidence}}
\label[appendix]{app:efx-proof}

\thmefxconfidence*

\begin{proof}
Fix $n$ and $\alpha$. For every $t\in[n]$, let $c_t=(\alpha^2/8)^{n-t}$.
Before the run, choose a permutation $(\pi_1,\dots,\pi_n)$ of the agents uniformly at random.
For $t,r\in[n]$, set $v_{\pi_r}(g_t)=c_t$ if $t\le r$, and $v_{\pi_r}(g_t)=\alpha c_r/4$ if
$t>r$. Give $g_{n+1}$ the same valuation vector as $g_n$. These choices specify the entire input.
For $t<n$, the agents $\pi_t,\dots,\pi_n$ have the same value for each of
$g_1,\dots,g_t$, and $\pi_t$'s value first decreases at $g_{t+1}$.

In any $\alpha$-EFX allocation of these $n+1$ strictly positive goods, every agent must
receive a good. Otherwise, an agent with an empty bundle envies a bundle containing at least
two goods after either of them is removed. Thus exactly one bundle contains two goods and all
others contain one.

We show that every successful allocation must give $g_t$ to $\pi_t$ for all $t<n$.
Suppose otherwise, and let $t<n$ be the smallest index such that $\pi_t$ does not receive
$g_t$. Agent $\pi_t$ receives none of the
first $t$ goods. Each later good is worth $\alpha c_t/4$ to her, and she receives at most
two goods, so $v_{\pi_t}(A_{\pi_t})\le\alpha c_t/2<\alpha c_t$. Consequently, the recipient of $g_t$ cannot receive any other good: removing that other
good would leave value $c_t$ to $\pi_t$. Since the earlier goods went to
$\pi_1,\dots,\pi_{t-1}$, the recipient of $g_t$ is some $\pi_r$ with $r>t$, and her
final utility is $v_{\pi_r}(A_{\pi_r})=c_t$.

Every later good $g$ has value at least $\alpha c_{t+1}/4$ to $\pi_r$. Indeed, its value
is at least $c_{t+1}$ until her value decreases, and is $\alpha c_r/4\ge\alpha c_{t+1}/4$
thereafter, since $r\ge t+1$. Hence
$\alpha v_{\pi_r}(g)\ge\alpha^2c_{t+1}/4=2c_t>v_{\pi_r}(A_{\pi_r})$. The first $t$ goods have distinct recipients, and there are $n+1-t$ goods still to arrive
but only $n-t$ empty bundles. Some later good must therefore share a bundle with another
good. This bundle cannot belong to $\pi_r$, whose good must remain alone. Removing the
other good from that bundle contradicts this inequality and $\alpha$-EFX.
This proves the required assignments of the first $n-1$ goods.

Fix a deterministic rule. Conditional on its earlier required assignments being correct,
$\pi_t$ is uniform among the $n-t+1$ agents whose values have not yet decreased. Their
revealed values do not distinguish them. The probability that all required assignments
are correct is therefore at most $\prod_{t=1}^{n-1}1/(n-t+1)=1/n!$.
Averaging this bound over the internal randomness of $\calA$ gives the same bound for $\calA$
on the random input. Hence, some fixed input has success probability at most $1/n!$, proving
the bound in \cref{thm:efx-confidence}. Knowing the common number of goods does not reveal the permutation.

For the positive guarantee, use the counting argument of
\citet[Proposition~3.6]{neohTeh2025efx}. Every input with $n+1$ goods has at least $n$ exact EFX
allocations with one two-good bundle and $n-1$ singleton bundles. To see this, choose any
agent to receive the pair. With all values known, let the other agents, in any order,
each choose a most valuable remaining good, and give the last two goods to the chosen
agent. Each of the other agents values her own good at least as much as either good in the
pair. Different choices of the pair's recipient give different allocations.

There are $\binom{n+1}{2}n!$ allocations of this form. An online algorithm can choose one
uniformly in advance as a sequence of recipients for the first $n+1$ arrival positions,
and allocate any later goods arbitrarily. On every input with $n+1$ goods, its exact EFX
success probability is at least $n/(\binom{n+1}{2}n!)=2/((n+1)n!)$. The algorithm uses only $n$, not advance knowledge of the number of goods.
\end{proof}

The same bounds hold for the optimal worst-case expected realized EFX guarantee restricted
to $n+1$ goods:
\begin{equation}
 \frac{2}{(n+1)n!}
 \le\sup_{\calA}\inf_{I:\,m=n+1}\E[\rho_{\text{EFX}}(\calA,I)]
 \le\frac1{n!}.
 \label{eq:efx-expected-short}
\end{equation}
The lower bound follows from the exact EFX probability in \cref{thm:efx-confidence}.
For the upper bound, fix $\calA$. For each $\alpha\in(0,1)$, the theorem gives an input with
expected ratio at most $\alpha+(1-\alpha)/n!$. As in \cref{cor:ef1-expected}, take the
infimum over inputs before letting $\alpha$ tend to zero. By Stirling's formula, both the
optimal success probability for each positive factor and the expected guarantee in
\eqref{eq:efx-expected-short} are $\exp(-n\ln n+n+O(\ln n))$.

\subsection{Proof of Theorem~\ref{thm:efx-known-max}}
\label[appendix]{app:efx-known-max}

\thmefxknownmax*

\begin{proof}
Let $q=\lfloor n/2\rfloor$. Before the run, choose a uniform $q$-element subset
$S\subseteq[n]$, and present the following $n+1$ goods in the indicated order.
\begin{center}
\begin{tabular}{@{}lccc@{}}
\toprule
Goods & Number & Value to $i\in S$ & Value to $i\notin S$\\
\midrule
Initial goods & $q$ & $\alpha/4$ & $\alpha/4$\\
Middle goods & $n-q$ & $\alpha^2/16$ & $1$\\
Final good & $1$ & $1$ & $1$\\
\bottomrule
\end{tabular}
\end{center}
The final good makes every maximum value equal to one. Neither the initial valuation
vectors nor the announced maxima reveal $S$.

Fix a deterministic rule, and let $C$ be the set of recipients of the first $q$ goods.
This set is independent of $S$ and has size at most $q$. We show that success requires
$|S\setminus C|\le1$. Suppose instead that at least two members of $S$ receive no initial good. At most one
receives the final good, so choose another such agent $i$. As in the proof of
\cref{thm:efx-confidence}, every $\alpha$-EFX allocation of these $n+1$ strictly positive
goods has one two-good bundle and $n-1$ singleton bundles. All goods that $i$ can receive
are middle goods and she receives at most two goods, so $v_i(A_i)\le\alpha^2/8<\alpha^2/4$. Every initial good must remain alone in its recipient's bundle: removing any accompanying
good would leave value $\alpha/4$ to $i$, contradicting $\alpha$-EFX. The final good must
also remain alone, since $i$ values it at one and does not receive it.

In particular, the initial recipients are distinct, so $|C|=q$. Because $S\setminus C$
is nonempty and $|S|=|C|$, choose an agent $a\in C\setminus S$. Her final utility is
$v_a(A_a)=\alpha/4$. The $n-q$ middle goods must all go outside $C$ and outside the final
good's recipient. There are only $n-q-1$ such agents, so some bundle must contain two
middle goods. Agent $a$ values each at one, contradicting $v_a(A_a)=\alpha/4<\alpha$.
This proves that success requires $|S\setminus C|\le1$.

If $|C|=q$, exactly $1+q(n-q)$ sets $S$ satisfy
$|S\setminus C|\le1$: either $S=C$, or one member of $C$ is replaced by
one agent outside $C$. If $|C|=q-1$, there are only $n-q+1$ possibilities, and if
$|C|\le q-2$, there are none. In every case, at most $1+q(n-q)$ of the $\binom nq$
possible sets can permit success. Dividing by $\binom nq$, averaging over the random choices of $\calA$,
and choosing a fixed input proves \eqref{eq:efx-known-max}. The argument allows any
allocation of the remaining goods after the first $q$ arrivals, even if all their values
are revealed at that point.
\end{proof}

\section{Omitted Proofs in Section~\ref{sec:efk}}
\label[appendix]{app:sec-efk}

\subsection{Proof of Theorem~\ref{thm:efk}}
\label[appendix]{app:efk-proof}

\thmefk*

\begin{proof}
Fix $i\ne j$ and put $\alpha=1-\eps$. Order the goods that $i$ values positively by decreasing
value to $i$, breaking ties by arrival index. Write the resulting values as
$a_1\ge\cdots\ge a_s>0$, and let $N_i(r)$ and $N_j(r)$ count the goods assigned to $i$ and
$j$ among the first $r$ positions. Both counts are zero at $r=0$. The case $s=0$ is immediate,
so assume $s>0$.

This order is used only in the proof. Because the allocations are independent on a fixed
input, we may reveal their recipients in this order. For the current good, let $p_i$ and
$p_j$ be its probabilities of going to $i$ and $j$. Under $\Like$, $p_i=1/|N_t|$ and $p_j$ is
either $1/|N_t|$ or zero; under $\Rand$, both equal $1/n$. In either case, $p_i\ge p_j$.
Set $S_r=N_j(r)-N_i(r)/\alpha$ and $z=\eps/2$.
Conditional on the recipients already revealed, the increment is $1$, $-1/\alpha$, or zero,
and its exponential satisfies
\begin{align*}
 \E\left[e^{z(S_r-S_{r-1})}\mid S_0,\dots,S_{r-1}\right]
 &=1+p_j(e^z-1)+p_i(e^{-z/\alpha}-1)\\
 &\le1+p_j(e^z+e^{-z/\alpha}-2)\le1.
\end{align*}
For the last inequality, $e^z\le(1-z)^{-1}$ and $e^{-z/\alpha}\le(1+z/\alpha)^{-1}$ give
$e^z+e^{-z/\alpha}\le(1-z)^{-1}+\alpha/(\alpha+z)=(1+\alpha)/(1-z)=2$, where $\alpha+z=1-z$ follows from $\alpha=1-\eps$ and $z=\eps/2$.

Thus $e^{zS_r}$ has conditional expectation at most its preceding value and begins at one.
Stop at the first position with $S_r>k$, or at $s$ if there is no such position. The stopped
sequence satisfies the same conditional expectation inequality, so its final expectation is
at most one. On the event that $S_r>k$ at some position, its final value exceeds $e^{zk}$.
Hence, by Markov's inequality, some $r$ has $S_r>k$ with probability at most
$e^{-zk}=e^{-\eps k/2}$. Outside this event, every position satisfies $N_i(r)\ge\alpha\max\{0,N_j(r)-k\}$.

Remove the $k$ most valuable goods, according to $i$, from $A_j$, or all its goods if there
are fewer than $k$, with ties broken as above. Among the first $r$ positions, exactly
$\max\{0,N_j(r)-k\}$ goods in $A_j$ remain. Set $a_{s+1}=0$. Summation by parts gives
\begin{align*}
 v_i(A_i)&=\sum_{r=1}^s(a_r-a_{r+1})N_i(r),\\
 v_i(A_j)-b_{ij}^{(k)}
 &=\sum_{r=1}^s(a_r-a_{r+1})\max\{0,N_j(r)-k\}.
\end{align*}
All coefficients are nonnegative, so $v_i(A_i)\ge\alpha(v_i(A_j)-b_{ij}^{(k)})$ except with
probability $e^{-\eps k/2}$. A union bound over the $n(n-1)$ ordered pairs proves
\eqref{eq:efk-probability} for both algorithms.
\end{proof}

\subsection{Proof of Proposition~\ref{prop:like-efk-lower}}
\label[appendix]{app:efk-like-lower}

\proplikeefklower*

\begin{proof}
Partition $2\lfloor n/2\rfloor$ agents into disjoint pairs. For each pair, present
$m_0=2\lceil k/\eps\rceil$ goods worth one to its two members and zero to everyone else.
Any remaining agent values all goods at zero. These choices specify a fixed input, and
$\Like$ allocates the goods independently within each pair.

For one pair, let $X$ be the number of goods given to its first agent, so that $X$ has
distribution $P=\operatorname{Bin}(m_0,1/2)$. Since this agent values each of the other
member's $m_0-X$ goods at one, she violates $(1-\eps)$-EF$k$ toward the other member when
$X<(1-\eps)(m_0-X-k)$, that is, when $X<(1-\eps)(m_0-k)/(2-\eps)$. Since $m_0\ge2k/\eps$, this threshold is at least
$(1-\eps)m_0/2$. We show that $X$ lies below the threshold with probability at least
$\tfrac12e^{-13\eps k}$.

Let $Q=\operatorname{Bin}(m_0,1/2-\eps)$. We compare the probabilities under $P$ and $Q$
of the event $\calE=\{|X-(1/2-\eps)m_0|\le\eps m_0/4\}$. Its upper endpoint is
$(1/2-3\eps/4)m_0$, strictly below $(1-\eps)m_0/2$ and hence below the threshold. Also, $\eps^2m_0\ge2\eps k\ge8$, so Chebyshev's
inequality gives $Q(\calE)\ge1-4/(\eps^2m_0)\ge1/2$. For any integer $x\in\calE$, comparing the binomial probabilities and using
$\ln u\le u-1$ gives
\begin{align*}
 \ln\frac{P(X=x)}{Q(X=x)}
 &=-x\ln(1-2\eps)-(m_0-x)\ln(1+2\eps)\\
 &\ge4\eps(x-m_0/2)\\
 &\ge-5\eps^2m_0\ge-13\eps k.
\end{align*}
The final inequality follows from $m_0<2k/\eps+2$, $\eps\le1/4$, and $\eps k\ge4$.
Summing over the event yields
$P(\calE)\ge e^{-13\eps k}Q(\calE)\ge\tfrac12e^{-13\eps k}$.

A $(1-\eps)$-EF$k$ allocation must satisfy the comparison within every pair. The
allocations of distinct pairs are independent, so its probability is at most
\[
 \left(1-\tfrac12e^{-13\eps k}\right)^{\lfloor n/2\rfloor}
 \le\exp\left(-\tfrac12\lfloor n/2\rfloor e^{-13\eps k}\right),
\]
which proves \eqref{eq:efk-like-lower}. For $\Like$ to guarantee success probability at least $1-\delta$ on every input, this upper
bound must be at least $1-\delta$. Hence,
\[
 \tfrac12\lfloor n/2\rfloor e^{-13\eps k}
 \le -\ln(1-\delta)\le2\delta
\]
for $\delta\le1/2$. Rearranging gives the necessary bound on $k$.
\end{proof}

\subsection{Proof of Corollary~\ref{cor:efk-expected}}
\label[appendix]{app:efk-expected}

\corefkexpected*

\begin{proof}
For every input, \eqref{eq:efk-probability} gives
\begin{align*}
 \E[1-\rho_{\text{EF}k}(\Like,I)]
 &=\int_0^1\Prb[\rho_{\text{EF}k}(\Like,I)<1-\eps]\,d\eps\\
 &\le\int_0^1\min\{1,n(n-1)e^{-\eps k/2}\}\,d\eps\\
 &\le\frac{2(\ln(n(n-1))+1)}k.
\end{align*}
For the last inequality, split the integral at $2\ln(n(n-1))/k$ if this point is below one;
otherwise the displayed upper bound is at least one. Taking the infimum of the expected
ratio over inputs proves the lower bound in \cref{cor:efk-expected}; the argument for $\Rand$ is identical.

For the reverse estimate in \cref{cor:efk-expected}, namely
$1-R^{\Like}_{\text{EF}k}(n)=\Omega(\ln n/k)$, set $\eps=\ln n/(26k)$ in
\cref{prop:like-efk-lower}. When $k\ge\ln n$ and $n$ is sufficiently large, its conditions
hold. Since $e^{-13\eps k}=n^{-1/2}$, \eqref{eq:efk-like-lower} shows that, on the resulting
input, $\Like$ returns a $(1-\eps)$-EF$k$ allocation with probability at most
$\exp(-\lfloor n/2\rfloor/(2\sqrt n))$. The ratio is less than $1-\eps$ otherwise and never
exceeds one, so
\[
 \E[\rho_{\text{EF}k}(\Like,I)]
 \le1-\frac{\ln n}{26k}
       \left(1-\exp\left(-\frac{\lfloor n/2\rfloor}{2\sqrt n}\right)\right).
\]
The parenthesized factor tends to one, proving the matching order.
\end{proof}

\end{document}